\documentclass[11pt]{article}
\usepackage{preamble}

\title{Tarski Lower Bounds from Multi-Dimensional Herringbones\footnote{This research was supported by US National Science Foundation CAREER grant CCF-2238372. Authors are in alphabetical order.}}

\author{Simina Br\^anzei\footnote{Purdue University. E-mail: simina.branzei@gmail.com.} \and Reed Phillips\footnote{Purdue University. E-mail: phill289@purdue.edu.} \and Nicholas Recker\footnote{Epic. E-mail: nrecker@umich.edu.}}

\begin{document}
\maketitle 

\thispagestyle{empty}

\begin{abstract}
 Tarski's theorem  states that every monotone function from a  complete lattice to itself has a fixed point. 
We analyze the query complexity of finding such a fixed point on the 
$k$-dimensional grid of side length $n$ under the 
$\leq$ relation. In this setting, there is an unknown monotone function $f: \{0,1,\ldots, n-1\}^k \to \{0,1,\ldots, n-1\}^k$ and  an algorithm must query a vertex $v$ to learn  $f(v)$. The goal is to find a fixed point of $f$ using as few oracle queries as possible.

We show that the randomized query complexity of this problem is  $\Omega\left( \frac{k \cdot \log^2{n}}{\log{k}} \right)$ for all $n,k \geq 2$. This  unifies and improves upon two prior results: a lower bound of $\Omega(\log^2{n})$  from \cite{etessami2019tarski} and a lower bound of $\Omega\left( \frac{k \cdot \log{n}}{\log{k}}\right)$ from  \cite{BPR24}, respectively.
\end{abstract}

\newpage
\pagenumbering{arabic}

\section{Introduction}

Let $\mathcal{L}$ be a complete lattice under the $\leq$ relation. A function $f : \mathcal{L} \to \mathcal{L}$ is monotone if $x \leq y$ implies $f(x) \leq f(y)$ for all $x, y \in \mathcal{L}$. Tarski's theorem \cite{tarski1955lattice}  states that every such function has a fixed point.
Tarski's theorem has applications in many areas. In game theory, pure Nash equilibria in supermodular games can be viewed as fixed points of monotone functions, and hence by Tarski's theorem such equilibria exist \cite{etessami2019tarski}. In denotational semantics, the semantics of recursively defined programs can be characterized as the least fixed points of monotone functions, ensuring a well-defined interpretation of recursion \cite{Tarski_pl_note}.

Recently, the computational complexity of finding Tarski fixed points has been investigated for grids. Specifically, for $k,n \in \mathbb{N}$, consider the $k$-dimensional grid $\{0, \ldots, n-1\}^k$ of side length $n$. Let $\leq$ be the binary relation  where for each pair of vertices ${a} = (a_1, \ldots, a_k) \in \{0, \ldots, n-1\}^k$ and ${b} = (b_1, \ldots, b_k) \in \{0, \ldots, n-1\}^k$, we have  ${a} \leq {b}$ if and only if $a_i \leq b_i$ for each $i \in [k]$. 

In the query model,  there is an unknown monotone function $f:\{0, \ldots, n-1\}^k \to \{0, \ldots, n-1\}^k$. The problem $TARSKI(n, k)$ is to find a fixed point of $f$ given query access to the function, where the answer to each query $v \in \{0, \ldots, n-1\}^k$ is $f(v)$. The randomized query complexity is the minimum  expected number of queries required to find a solution with a probability at least $9/10$, where the expectation is taken over the coin tosses of the algorithm. 

There are two main algorithmic approaches for finding a Tarski fixed point on grids. One approach is a path-following method that uses $O(nk)$ queries, which is best for small $n$ and large $k$. For large $n$ and small $k$, a divide-and-conquer approach works much better. Originally proposed by \cite{dang2011computational} and subsequently improved by \cite{fearnley2022faster} and \cite{chen2022improved}, it achieves an upper bound of $O\left((\log{n})^{\lceil \frac{k+1}{2} \rceil}\right)$ queries for any constant $k$. 

There are a few known lower bounds which are optimized for different parameter regimes. \cite{etessami2019tarski} proved a randomized query complexity lower bound of $\Omega\left(\log^2(n)\right)$ for $TARSKI(n, 2)$. The same lower bound applies to $TARSKI(n, k)$ for $k \geq 2$. This lower bound matches the known upper bound for $k=2$ and $k=3$, though it does not scale with $k$.
\cite{BPR24} proved two dimension-dependent randomized  lower bounds, of $\Omega(k)$ and $\Omega\left(\frac{k \log n}{\log k}\right)$, respectively, for $TARSKI(n, k)$. These bounds scale with $k$, so they are stronger than \cite{etessami2019tarski} when $k$ grows at about $\text{polylog}(n)$. In particular, they characterize the randomized query complexity of finding a Tarski fixed point on the Boolean hypercube $\{0,1\}^k$ (i.e. the power set lattice) as $\Theta(k)$. However, they are weaker than \cite{etessami2019tarski} for very small $k$. 

In this paper, we show that the randomized query complexity of $TARSKI(n, k)$ is $\Omega\left(\frac{k \log^2(n)}{\log k}\right)$ for all $n,k \geq 2$.
This is a dimension-dependent lower bound that unifies and improves upon the  lower bounds of $\Omega\left(\log^2(n)\right)$  from \cite{etessami2019tarski} and  of $\Omega\left(\frac{k \log n}{\log k}\right)$ from \cite{BPR24}.

\subsection{Our Contribution}

Our main result is the following theorem.

\begin{restatable}{mytheorem}{theoremmainlowerbound}
\label{thm:lb_k_log_squared_n}
There exists $c > 0$ such that for all $k,n \in \mathbb{N}$ with $k, n \geq 2$, the randomized query complexity of $TARSKI(n,k)$ is at least $c \cdot  \frac{k \log^2{n}}{\log{k}}$.
\end{restatable}


\vspace{1mm}
\noindent \textbf{Proof Overview.} 
Our first contribution is to design a  multi-dimensional family $\mathcal{F}$ of herringbone functions, which generalizes the 2D construction from \cite{etessami2019tarski}. 

At a high level, 
there is a path (``spine'') that runs from vertex $(0, \ldots, 0)$ to $(n-1, \ldots, n-1)$ that contains the unique fixed point. Along the spine, function values flow directly towards the fixed point. Outside the spine, the function is designed to be consistent with any position of the fixed point along the spine, which forces the algorithm to find ``many" spine vertices to succeed. The challenge we overcome is defining such a function in a way that explicitly uses all the dimensions without violating monotonicity.

To obtain a lower bound, we invoke Yao's lemma, which allows focusing on a deterministic algorithm $\mathcal{A}$ that receives inputs drawn from the uniform distribution $\mathcal{U}$ over functions from $\mathcal{F}$. The high level idea  is that  distribution $\mathcal{U}$ has the property that the location of each spine vertex is independent of spine vertices more than a short distance away. 
Consequently, a successful algorithm must repeatedly find spine vertices without relying on previously discovered ones, and finding each new spine vertex requires ``many'' queries.

Formalizing this intuition requires careful handling. We design  a monotonic measure of progress that starts with a low value, increases by a small amount in expectation with each query, and must arrive at a high value for $\mathcal{A}$ to succeed. To do so, we first divide the space around the spine into regions. For each region $R$, we define random variables:
\begin{itemize}
\item $ P_t(R) = \max_{v \in R} \log_2 \Pr \left[\text{$v$ is a spine vertex after $t$ queries from $\mathcal{A}$}\right] $ for all $t \in \mathbb{N}$. 
\item $    P^*_t(R) = \max_{0 \leq t^* \leq t} P_{t^*}(R)  $ for all $t \in \mathbb{N}$. 
\end{itemize}
For $m \in \Theta(\log n)$, let $G_m$ be the set of $m$-tuples of regions that are sufficiently far apart from each other that the location of the spine in each region is independent of the others in the tuple.
Define 
\begin{itemize} 
\item $ \overline{P}_t = \max_{(R^1, \ldots, R^m) \in G_m} \sum_{i=1}^m P^*_t(R^i) $ for all $t \in \mathbb{N}$. 
\end{itemize}
The proof  shows that $\overline{P}$ is a monotonic measure of progress  with the required properties, which implies a lower bound on the number of queries needed for the algorithm to succeed. 

\vspace{-2mm}
\paragraph{Roadmap to the paper.} The remainder of the paper is organized as follows. Related work is in Section~\ref{sec:related_work}. The model and preliminaries are in Section \ref{sec:model}. 

Our family $\mathcal{F}$ of multi-dimensional herringbone functions can be found in Section~\ref{sec:multi_dim_herringbones}, with proofs of its properties in Appendix~\ref{app:properties_herringbone_functions}. 
In Section~\ref{sec:hard_distribution} we define a  distribution $\mathcal{U}$ on the family  of functions  $\mathcal{F}$, with proofs of its properties in Appendix~\ref{app:hard_distribution}.  The proof of Theorem~\ref{thm:lb_k_log_squared_n} can be found in Section~\ref{sec:proof_main_theorem}, with proofs of supporting lemmas in Appendix~\ref{app:proof_main_theorem}. 

In Section~\ref{sec:upper_bound_distribution_U} we also give upper bounds on the query complexity of finding a fixed point of any multi-dimensional herringbone function. 

\subsection{Related Work} \label{sec:related_work}

Algorithms for the problem of finding Tarski fixed points on the $k$-dimensional grid of side length $n$ have only recently been considered. \cite{dang2020tarskialgorithm} gave an $O(\log^k (n))$ divide-and-conquer algorithm. \cite{fearnley2022faster} gave an $O(\log^2 (n))$ algorithm for the 3D grid and used it to construct an $O(\log^{2 \lceil k/3 \rceil} (n))$ algorithm for the $k$-dimensional grid of side length $n$. \cite{chen2022improved} extended their ideas to get an $O(\log^{\lceil (k+1)/2 \rceil}(n))$ algorithm.

 \cite{etessami2019tarski} showed a lower bound of $\Omega(\log^2(n))$ for the 2D grid, implying the same lower bound for the $k$-dimensional grid of side length $n$. This bound is tight for $k = 2$ and $k=3$, but there is an exponential gap for larger $k$. They also showed that the problem is in both PLS and PPAD, which by the results of \cite{fearnley2022cls} implies it is in CLS.

\cite{CLY23} give a black-box reduction from the Tarski problem to the same problem with an additional promise that the input function has a unique fixed point. This result  implies that the Tarski problem and the unique Tarski problem have the same query complexity.

Two problems related conceptually to that of finding a Tarski fixed point are finding a Brouwer fixed point~\cite{hirsch1989exponential,chen2005algorithms,chen2007paths} and finding a local minimum (i.e. local search)~\cite{aldous1983minimization,Aaronson06,zhang2009tight,llewellyn1989local,sun2009quantum,santha2004quantum,dinh2010quantum}. The query complexity lower bounds for Brouwer and local search also rely on hidden path (``spine'') constructions. However, the monotonicity condition of the function in the Tarski setting poses an extra challenge.

\section{Model} \label{sec:model}

Let $\{0, 1, \ldots, n-1\}^k$ be the $k$-dimensional grid of side length $n$ under the $\leq$ relation.  Given oracle access to an (unknown) monotone function $f: \{0,1,\ldots, n-1\}^k \to \{0,1,\ldots, n-1\}^k$, the Tarski search problem,  ${TARSKI}(n,k)$, is to find a fixed point of $f$ using as few oracle queries as possible.

\begin{definition} [${TARSKI}(n,k)$]
Let $k,n \in \mathbb{N}$.
     Given oracle access to an unknown monotone function $f : \{0, 1, \ldots, n-1\}^k \to \{0, 1, \ldots, n-1\}^k$, find a vertex $x \in \{0, 1, \ldots, n-1\}^k$ with $f(x) = x$ using as few queries as possible.
 \end{definition}

\paragraph{Query complexity.} The deterministic query complexity of a task is the total number of queries necessary and sufficient for a correct deterministic algorithm to find a solution. The randomized query complexity is the expected number of queries required to find a solution with probability at least 9/10 for each input \footnote{Any other constant greater than $1/2$ would suffice.}, where the expectation is taken over the coin tosses of the algorithm.

Given an algorithm $\mathcal{A}$ for $TARSKI(n,k)$ that receives an input drawn from some input  distribution $\mathcal{D}$, we say that an execution of $\mathcal{A}$ \emph{succeeds} if it outputs a fixed point of the function given as input and \emph{fails} otherwise.

\medskip 

\noindent \textbf{Notation.} The following notation is used throughout the paper. 
\begin{itemize}
    \item For a vector $\vec{x}$ in $\mathbb{R}^k$ and $i \in [k]$, we write the $i$-th coordinate of $\vec{x}$ as $x_i$.
    \item For a vector $\vec{x} \in \mathbb{R}^k$, we use $wt(\vec{x})$ to denote its Hamming weight $\sum_{i=1}^k x_i$.
    \item For each $j \in [k]$, let
        $\vec{e}_j = (0, \ldots, 0, 1, 0, \ldots, 0)$,
    where the vector $\vec{e}_j$ is all zeroes except the $j$-th coordinate in which it takes value $1$.
    \item For all $u, v \in \mathbb{R}^k$,  the line segment between them is 
        $\ell(u, v) = \bigl\{ \lambda u + (1-\lambda)v \mid \lambda \in [0, 1] \bigr\} \,. $ 
\end{itemize}


\section{Multi-dimensional Herringbone Functions} \label{sec:multi_dim_herringbones}

Our first contribution is to generalize the lower bound  construction from \cite{etessami2019tarski} to any number of dimensions. Towards this end, we need to state the definition of a spine, which is 
a monotone sequence of vertices that each have Hamming distance $1$ from their neighbors.


\begin{definition}[Spine] \label{def:spine}
A spine $\mathbf{s} = \{s^0, \ldots, s^{(n-1)k}\}$ is a sequence of vertices with the property that $s^0 = (0, \ldots, 0)$, $s^{(n-1)k} = (n-1, \ldots, n-1)$, and $s^{i+1} > s^i$ for all $i \in \{0, \ldots, (n-1)k-1\}$.
\end{definition}

\begin{lemma} \label{lem:weight_of_spine_vertex}
    Let $\vec{s} = (s^0, \ldots, s^{(n-1)k})$ be a spine. Then $wt(s^i) = i$ for each $ i \in \{0, \ldots, (n-1)k\}$.
\end{lemma}
\begin{proof}
    By the definition of a spine, we have $s^{i+1} > s^i$ for all $i \in \{0, \ldots, (n-1)k-1\}$. Thus $wt(s^{i+1}) > wt(s^i)$. The available Hamming weights are $\{0, \ldots, k(n-1)\}$, thus $wt(s^i) = i$ for all $i \in \{0, \ldots, (n-1)k\}$.
\end{proof}

\begin{definition}[Multi-dimensional Herringbone Function]
    \label{def:general-herringbone}
    Consider a spine $\mathbf{s} = \{s^0, \ldots, s^{(n-1)k}\}$ and an index  $j \in \{0, \ldots, (n-1)k\}$.
Define functions $M, \mu : \{0, \ldots, n-1\}^k \to \{0, \ldots, (n-1)k\}$ as: 
    \begin{align}
        M(v) = \min \left\{i \in \{0, \ldots, (n-1)k\} \mid s^i \geq v \right\} \; \; \;   \mbox{and} \; \; \;    
        \mu(v) = \max \left\{ i  \in \{0, \ldots, (n-1)k\} \mid  s^i \leq v \right\}\,.  \notag 
    \end{align}
    Let $h^{\mathbf{s}, j} : \{0, \ldots, n-1\}^k \to \{0, \ldots, n-1\}^k$  be the function:
    \begin{equation}
        h^{\mathbf{s}, j}(v) = \begin{cases}
            v & \text{if } v = s^j \\
            s^{i+1} & \text{if } v = s^i \text{ for some } i < j \\
            s^{i-1} & \text{if } v = s^i \text{ for some } i > j \\
            v + (s^{\mu(v)+1} - s^{\mu(v)}) - (s^{M(v)} - s^{M(v)-1}) & \text{otherwise} \,.
        \end{cases}
    \end{equation}
\end{definition}

The first three cases of \cref{def:general-herringbone} are the same as in the two-dimensional construction of \cite{etessami2019tarski}, as they naturally extend to the higher-dimensional setting without modification. 

Our key innovation lies in the fourth case, which applies to vertices not on the spine. It is designed to hide the location of $s^j$ by being consistent with any value of $j$.

We show that each function $h^{\mathbf{s}, j} $ from  \cref{def:general-herringbone} has a unique fixed point at $s^j$ and is monotone. The proofs are included in Appendix~\ref{app:properties_herringbone_functions}.
\begin{restatable}{mylemma}{lemmauniquefixedpoint}
\label{lem:h-unique-fixedpoint}
    The function $h^{\mathbf{s}, j}$ in Definition \ref{def:general-herringbone} has a unique fixed point at $s^j$.
\end{restatable}

\begin{restatable}{mylemma}{lemmamonotone}
\label{lem:h-monotone}
    The function $h^{\mathbf{s}, j}$ in Definition \ref{def:general-herringbone} is monotone.
\end{restatable}

On the spine, $h^{\mathbf{s}, j}$ follows the spine to $s^j$. Off the spine, $h^{\mathbf{s}, j}(v)$ increases in the dimension that the spine moves after $s^{\mu(v)}$ and decreases in the dimension that the spine moves before $s^{M(v)}$. 

In two dimensions, this coincides with the herringbone function of \cite{etessami2019tarski}: vertices above the spine go down-right, and vertices below the spine go up-left. In higher dimensions, however, the off-spine vertices do not in general point directly at the spine.

\section{Hard Distribution} \label{sec:hard_distribution}

In this section we define a hard distribution of inputs that is used to prove Theorem~\ref{thm:lb_k_log_squared_n}. We will use multi-dimensional herringbone functions with spines of a particular form. The proofs of statements in this section can be found in Appendix~\ref{app:hard_distribution}.

The idea is to restrict the spine to a ``tube" running from the minimum to the maximum vertex in the lattice with width $\text{poly}(n)$. We select connecting points on several slices throughout the tube and make the spine interpolate between these points.

We first introduce the tube and some of its useful properties.

\begin{definition}[Tube]
    The tube of width $L$ in $\{0, \ldots, n-1\}^k$ is the set of points:
    \begin{equation}
        T_L = \{ x \in \{0, \ldots, n-1\}^k \mid  \exists i \in \{0, \ldots, n-1\} \text{ such that $|x_j - i| \leq L$} \;  \forall j \in [k]\}\,. \label{eq:condition_to_be_in_tube}
    \end{equation}
\end{definition}

The main appeal of the tube is that queries inside the tube can only provide information about nearby spine vertices.

\begin{restatable}{mylemma}{lemmatubequerylocality}
\label{lem:tube-query-locality}
     For a point $v \in T_L$ and a multi-dimensional herringbone function $f$ whose spine lies in $T_L$, the value of $f(v)$ depends only on spine vertices whose coordinates are all within $2L$ of $v$'s.
\end{restatable}

However, an algorithm is free to query any vertices, not just ones in the tube. We get around this by reducing a query to a vertex $v$ outside the tube to two queries on the boundary of the tube, which together provide the same information as querying $v$.

\begin{restatable}{mylemma}{lemmaoutsidequerysimulation}
\label{lem:outside-query-simulation}
 Let $L \in \mathbb{N}$ and  $f: \{0, \ldots, n-1\}^k \to \{0, \ldots, n-1\}^k$  be a  multi-dimensional herringbone function whose spine lies entirely in $T_L$. 
    For each  point $v \not \in T_L$, there exist  $a, b \in T_L$ such that:
    \begin{itemize}
        \item $a$ and $b$ can be computed from $v$ and $L$; and
        \item $f(v)$ can be computed from $f(a)$ and $f(b)$.
    \end{itemize}
\end{restatable}


We now divide the tube $T_L$ into regions, whose  sizes are defined using a parameter $\rho  \geq 0$ such that $\rho \mid k(n-1)$, $ k \mid \rho$, and $\rho \geq 12kL$. We will eventually use $L = \frac{1}{12} \sqrt{n-1}$ and $\rho = k \sqrt{n-1}$, but these precise values are not important until the final bound is evaluated. We do require $L$ and $\rho$ to be integers, which will only occur if $n-1$ is the square of an integer that is divisible by $12$. However, such values of $n$ are sufficiently common that rounding down to the nearest one is a negligible loss.

We define a set of indices $a$ for the regions we focus on as follows:
\begin{align}
    \mathcal{I} =  \Bigl\{a \mid a \in \{0, \ldots, k(n-1)\} \; \mbox{and} \; \rho \mid a \Bigr\}\,.
\end{align}

\begin{definition}[Regions]
For each $a \in \mathcal{I} \setminus \{k(n-1)\}$, define the region 
\begin{align}
    R_a = \{ \vec{x} \in T_L \mid  a \leq wt(\vec{x}) \leq a + \rho \}\,. \label{eq:region_R_a}
\end{align}
For all $a \in \mathcal{I}$, let 
    $Low_a = \{ \vec{x} \in T_L \mid wt(\vec{x}) = a\}\,.$
For each region $R_a$, $Low_a$ is its lower boundary.
\end{definition}

The spines of our hard input distributions will be defined by passing through a series of connecting points, one at each region boundary.
\begin{definition}[Connecting points]
Given $i \in \mathcal{I}$, the  vertices $\chi_i$ are called \emph{connecting points} if $\chi_i \in Low_i$.
\end{definition}
Given a sequence of connecting points $\chi_i$, we next construct a spine that interpolates between them. Each $\chi_i$ represents the vertex that the spine passes through when entering region $R_{i}$ from $R_{i-\rho}$. 

For each vertex $v \in \{0, \ldots, n-1\}^k$, let $\mathcal{C}(v)$ be the axis-aligned cube of side length $1$ centered at $v$, including its boundary faces. 
    For each  $i \in [k]$, the face of $\mathcal{C}(v)$ with minimum value in coordinate $i$ is called \emph{backward} and the face  with maximum value in coordinate $i$ is called \emph{forward}.

\begin{restatable}{mylemma}{lemhelpdefinitionspine}
\label{lem:help_definition_spine}
Let $a,b \in \{0,\ldots, n-1\}^k$ with $a \leq b$. 
For an arbitrary vertex $v \in \{0, \ldots, n-1\}^k$, suppose the cube $\mathcal{C}(v)$ intersects the line segment $\ell(a, b)$. Let $z \in \mathbb{R}^k$ be such that   $z = \max \{\mathcal{C}(v) \cap \ell(a, b)\}$. Then $z$ is well-defined, and either $v=b$ or $z$ lies on some forward face of $\mathcal{C}(v)$.
\end{restatable}

Using Lemma \ref{lem:help_definition_spine} as a subroutine, we can construct a path that follows the continuous line segment $\ell(u, v)$ between two given vertices $u$ and $v$ as closely as possible. In particular, each vertex $s^i$ on the path will include some of the points in  $\ell(u, v)$ in its cube $\mathcal{C}(s^i)$.

\begin{restatable}{mylemma}{lemsequencesjbetweenchis}
\label{def:sequence_s_j_between_chi_i_and_chi_i_plus_rho}
Let   $u,v \in \{0, \ldots, n-1\}^k$ be vertices with $u \leq v$. Then there exists a sequence of vertices $s^1, \ldots, s^m$ for some $m \in \mathbb{N}$ such that  
\begin{enumerate}
    \item $s^1 = u$ and $s^m = v$.
    \item $(s^1, \ldots, s^m)$ is a monotone connected path in the graph $\{0, \ldots, n-1\}^k$.
    \item For each $i \in [m]$, the set $\mathcal{C}(s^i) \cap \ell(u, v)$ is nonempty.
\end{enumerate}
\end{restatable}
An illustration of a sequence given by Lemma~\ref{def:sequence_s_j_between_chi_i_and_chi_i_plus_rho} can be found in Figure~\ref{fig:line-through-grid_main}.

\begin{figure}[h!]
\centering 
\includegraphics[scale=0.36]{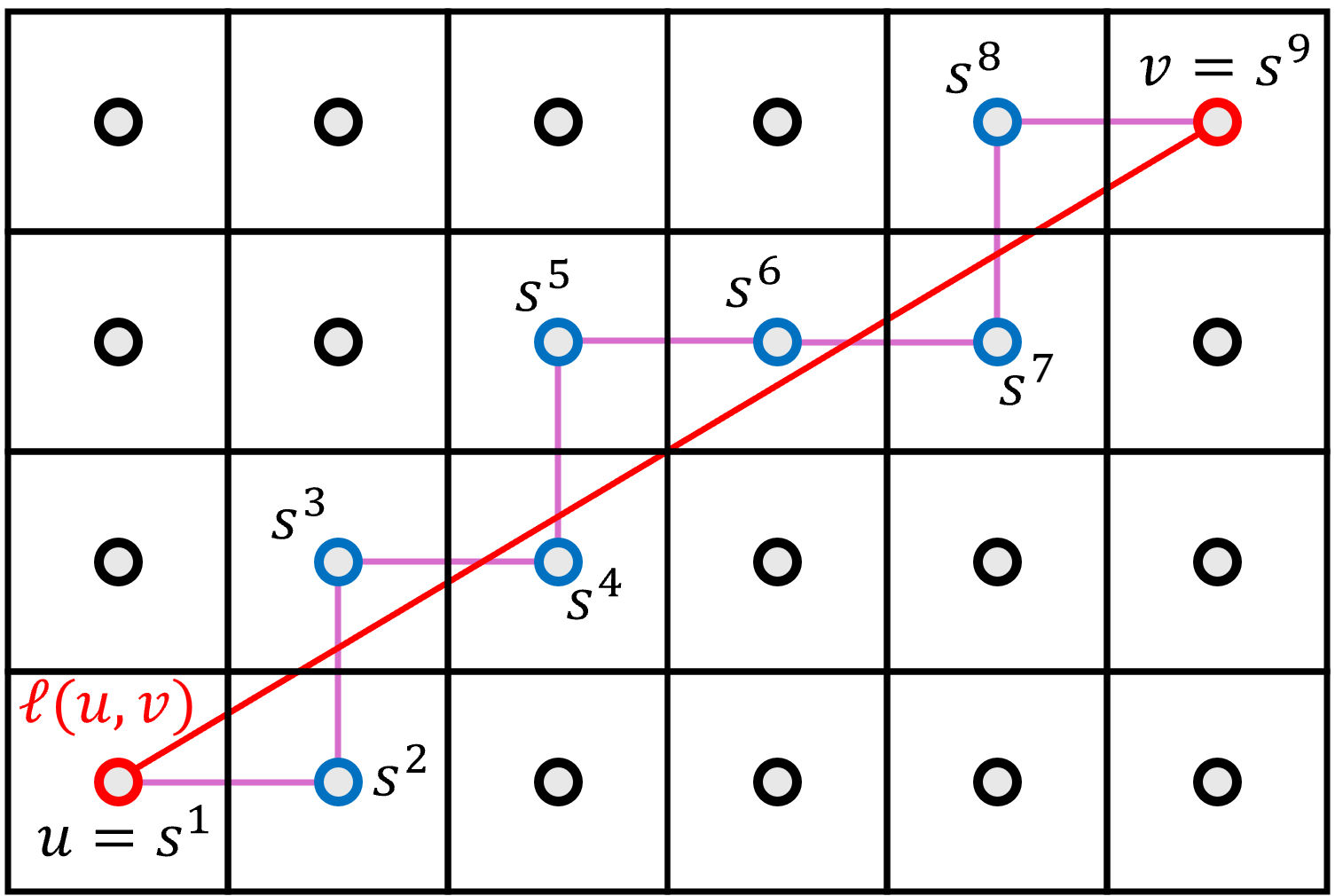}
\caption{A 2-dimensional example of a sequence given by Lemma \ref{def:sequence_s_j_between_chi_i_and_chi_i_plus_rho}, with endpoints $u$ and $v$ and the associated spine vertices $s^1$ through $s^9$. Each vertex is drawn surrounded by its cube.}
\label{fig:line-through-grid_main}
\end{figure}

A sequence of vertices $s = (s^1, \dots, s^q)$ is \emph{lexicographically smallest} among a set $\mathcal{W}$ of sequences if $s$ has the smallest (lexicographically) possible first vertex $s^1$ among all choices for $s^1$ in $\mathcal{W}$, and among sequences tied in the first vertex, the lexicographically smallest possible second vertex $s^2$, and so forth.

\begin{definition}[Spine induced by connecting points] \label{def:connecting-point-spine}
Let $\chi_i$ for $i \in \mathcal{I}$ be connecting points. 
For each consecutive pair of connecting points $\chi_i$ and $\chi_{i+ \rho}$, define the sequence  $(s^i, s^{i+1}, \ldots, s^{i+\rho})$ as the lexicographically smallest one that satisfies the properties of Lemma~\ref{def:sequence_s_j_between_chi_i_and_chi_i_plus_rho} invoked with $u = \chi_i$ and $v = \chi_{i+\rho}$. Let the spine be  $\vec{s} = (s^0, \ldots, s^{(n-1)k})$.
\end{definition}

\begin{figure}[h!]
\centering 
\includegraphics[scale=0.4]{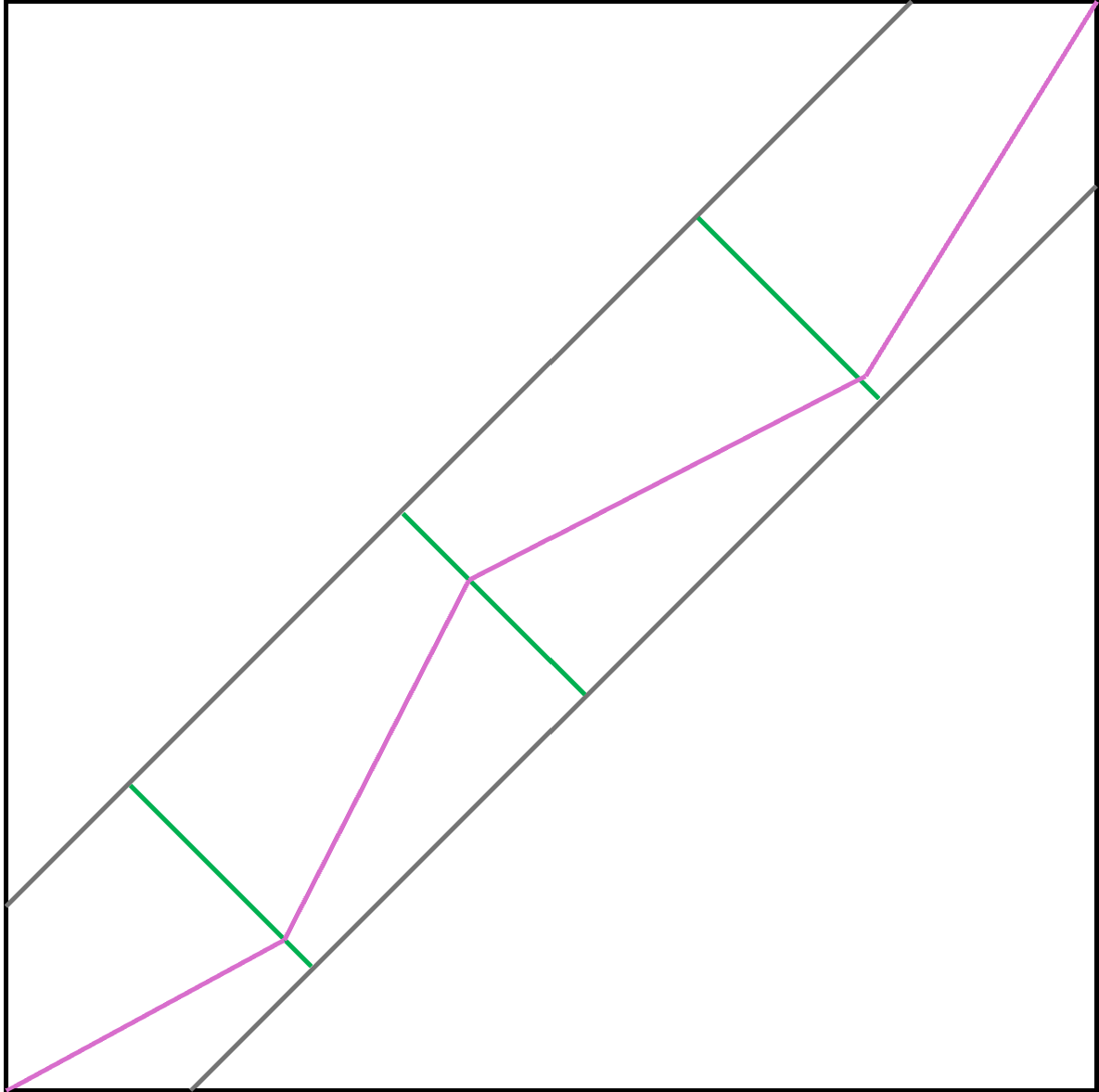}
\caption{A sketch of the structure of the spines in Definition \ref{def:connecting-point-spine}. The gray lines represent the edges of the tube $T_L$. The green lines represent $Low_a$ for various $a \in \mathcal{I}$, and divide $T_L$ into several regions. The pink path represents a possible spine through a choice of connecting points. }
\end{figure}

\begin{observation} \label{lem:spine_monotone}
Each spine $\vec{s}$ from Definition~\ref{def:connecting-point-spine} is monotone. 
\end{observation}
\begin{proof}
Given a set of connecting points $\chi_i$ for $i \in \mathcal{I}$, 
the spine from Definition~\ref{def:connecting-point-spine} is obtained  by generating a monotone path $P_i$ to connect each consecutive pair of connecting points $\chi_i$ and $\chi_{i+ \rho}$. Then the spine is the union of the $P_i$'s and thus monotone as well.
\end{proof}

\begin{definition}[Set of spines $\Psi$, family of functions $\mathcal{F}$, and distribution $\mathcal{U}$]
The set of spines we consider is
\begin{align}
    \Psi &= \Bigl\{\vec{s} \mid \exists \mbox{ connecting points } \chi_i \mbox{ for } i \in \mathcal{I} \mbox{ such that } \vec{s} \mbox{ is the spine induced by them}\Bigr\}
\end{align}
The family of functions we consider is 
    $ \mathcal{F}  = \bigl\{ h^{\vec{s}, j} \mid \vec{s} \in \Psi, j \in \{0, \ldots, k(n-1)\}\bigr\}$. 
Let   $\mathcal{U}$ be the uniform distribution over functions in $\mathcal{F}$.
\end{definition}

Choosing spine vertices independently in different regions allows us to show that queries within the tube truly give only local information, extending Lemma \ref{lem:tube-query-locality}.

\begin{restatable}{mylemma}{lemqueryaffectsadjacentonly}
\label{lem:query-only-affects-adjacent-regions}
    Suppose a vertex $v \in \{0, \ldots, n-1\}^k$ is in a region $R_{\alpha}$. On input $f \sim \mathcal{U}$, the value of $f(v)$ is independent of the location of spine vertices in all regions except possibly regions 
    $R_{\alpha + i\rho}$ for  $i \in \{-2, \ldots, 2\}$.
\end{restatable}


\section{Proof of the Main Theorem} \label{sec:proof_main_theorem}

We now move towards the proof of Theorem \ref{thm:lb_k_log_squared_n}.
At a high level, the lower bound in  this theorem comes from a combination of two ideas:
\begin{itemize}
    \item Lemma \ref{lem:many-far-apart-spine-vertices-hard}, which states that finding spine vertices in many different regions is difficult. Its proof is based on spine vertices being exponentially rare due to our construction.
    \item Lemma \ref{lem:many-region-queries-required}, which states that any correct algorithm for $TARSKI(n, k)$ must nevertheless find spine vertices in many different regions. Its proof is based on a reduction to ordered search.
\end{itemize}

The proofs of statements in this section can be found in Appendix~\ref{app:proof_main_theorem}.

We begin with the ideas leading to Lemma \ref{lem:many-far-apart-spine-vertices-hard}.  The general idea is that any given vertex is very unlikely to be a spine vertex (Lemmas \ref{lem:size_Low_a} and \ref{lem:random-spine-unlikely-to-hit-vertex}), while each query only provides $O(\log k)$ bits of information about the spine (Lemma \ref{lem:F-v-size-less-than-k-cubed}).

We show that the number of points on each boundary between regions is exponentially large.

\begin{restatable}{mylemma}{lemmasizelowa}
\label{lem:size_Low_a}
    Let $a \in \mathcal{I} \setminus \{0, k(n-1)\}$. Then 
$|Low_a| \geq (2L+1)^{\lfloor k/2 \rfloor}\,.$
\end{restatable}

Because the number of points on each boundary is large, randomly selecting connecting points on the two ends of region doesn't concentrate the spine anywhere in that region.

\begin{restatable}{mylemma}{lemmarandomspineunlikelyhitvertex}
\label{lem:random-spine-unlikely-to-hit-vertex}
Let $w \in T_L$ be a vertex in the tube and 
    $a,b \in \mathcal{I} \setminus \{0, k(n-1)\}$ such that 
     $b - a \geq 12kL$.  
    For random vertices $U \sim \mathcal{U}(Low_a)$ and $V \sim \mathcal{U}(Low_b)$, we have
\begin{equation} \label{eq:w_hits_sequence_s}
        \Pr \Bigl[ w \in \vec{s}(U, V) \Bigr] \leq \frac{17^k}{(2L+1)^{\lfloor k/2 \rfloor}},
    \end{equation} 
    where $\vec{s}(U, V)$ is the lexicographically smallest connected monotone path from $U$ to $V$ that satisfies the properties of  Lemma \ref{def:sequence_s_j_between_chi_i_and_chi_i_plus_rho}.
\end{restatable}

The next lemma shows that for each vertex $v$, the number of possible responses to querying $v$ is $poly(k)$. This follows from Definition \ref{def:general-herringbone} and shows that   each query only provides $O(\log k)$ bits of information.

\begin{restatable}{mylemma}{lemmaFvsizelessthankcubed}
\label{lem:F-v-size-less-than-k-cubed}
    For each vertex $v \in \{0, \ldots, n-1\}^k$, let
    \begin{align}
        F(v) &= \left\{y \in \{0, \ldots, n-1\}^k \mid \exists \vec{s} \in \Psi, j \in \{0, \ldots, k(n-1)\} \mbox{ such that } h^{\vec{s}, j}(v) = y\right\} \,. 
    \end{align}
    Then, for $k \geq 2$, we have 
$        |F(v)| \leq k^3\,.$
\end{restatable}

Our central argument will revolve around finding spine vertices in ``far-apart" regions. To that end, we define a distance between regions.

\begin{definition}
    Given two regions $R_{\alpha}$ and $R_{\beta}$, their distance is defined as:
    $        dist(R_{\alpha}, R_{\beta}) = |\alpha-\beta|/\rho \,.$
\end{definition}

Next we define what it means for an algorithm to ``survey the spine''. 

\begin{definition}[Surveying the spine]
On an input function $f \in \mathcal{F}$,  an 
algorithm ${\mathcal{A}}$ is said to \emph{$m$-survey the spine of $f$} if the following condition is met:
\begin{itemize}
    \item the set of vertices queried by ${\mathcal{A}}$ contains  $m$ spine vertices $v_1, \ldots, v_m$ such that $dist(R^i, R^j) \geq 5$ for all $i \neq j$, where $R^{a}$ is the region containing $v_{a}$.
\end{itemize}
\end{definition}

Now we can bound the success probability of algorithms that manage to survey the spine within a ``few'' queries.

\begin{restatable}{mylemma}{lemmamanyfarapartspineverticeshard}
\label{lem:many-far-apart-spine-vertices-hard}
Let $m \in \mathbb{N}$. Let $\mathcal{A}$ be a deterministic algorithm that has oracle access to an unknown function  $f \in \mathcal{F}$, where $\mathcal{A}$ is said to \emph{succeed} on input $f$ if it $m$-surveys the spine of $f$. Then if $\mathcal{A}$ makes at most $T$ queries, its success probability on inputs drawn from $\mathcal{U}$ is upper bounded by:
\begin{align}
    \Pr_{f \sim \mathcal{U}} [\text{$\mathcal{A}$ succeeds in $T$ queries on input $f$}] &\leq \frac{15 T \log_2(k)}{(m-2) \cdot \left(\lfloor k/2 \rfloor \log_2(2L+1) - k \log_2{17}\right)} \,.
\end{align}
\end{restatable}
\begin{proof}[Proof sketch]
We provide a proof sketch here; the complete proof can be found in Appendix~\ref{app:proof_main_theorem}.

The main idea is to find a measure of progress and show  that the algorithm cannot progress too quickly. Let $T$ be the number of queries issued by algorithm $\mathcal{A}$. For each time $t \in \{0, \ldots, T\}$, we define a random variable $\overline{P}_t$ such that: 
\begin{itemize} 
\item $\overline{P}_0$ is ``small''
\item the expectation of $\overline{P}_{t+1} - \overline{P}_t$ is  bounded from above, and
\item the final value $\overline{P}_T$ must be ``large'' 
for $\mathcal{A}$ to have a high success probability. 
\end{itemize} 
These properties collectively show a lower bound on the number of time steps 
$T$ required for  algorithm $\mathcal{A}$ to succeed with high probability.

To make this intuition precise, we define for each region $R$ the following random variables:
\begin{itemize}
   \item  $P_t(R) = \max_{v \in R} \log_2 \Pr \left[\text{$v$ is a spine vertex after $t$ queries from $\mathcal{A}$}\right] \,.  $
   \item $ P^*_t(R) = \max_{0 \leq t^* \leq t} P_{t^*}(R) \,. $ 
\end{itemize}
Let 
$     G_m = \left\{(R^1, \ldots, R^m) \mid R^1, \ldots, R^m \mbox{ are regions, } dist(R^i, R^j) \geq 5 \mbox{ for all } i \neq j \in [m]\right\} \notag  \,.  $ 
That is, $G_m$ is the set of $m$-tuples of regions which $\mathcal{A}$ is trying to find spine vertices in. Finally, define 
\begin{itemize}
\item $ \overline{P}_t = \max_{(R^1, \ldots, R^m) \in G_m} \sum_{i=1}^m P^*_t(R^i) \,. $
\end{itemize}

The proof shows that  
$\mathbb{E} [\overline{P}_{t+1} - \overline{P}_t] \leq 15 \log_2(k)$ for  all $t$. 
This implies that if $\mathcal{A}$  makes at most $T$ queries, then 
$ \mathbb{E} [\overline{P}_T - \overline{P}_0]  \leq 15 T \log_2(k) \,.$

When $\mathcal{A}$ succeeds, it has found $m$ spine vertices in regions at least five apart, and therefore $\overline{P}_T = 0$. By Lemma \ref{lem:random-spine-unlikely-to-hit-vertex}, no vertex (other than those in the first and last regions, which we can effectively ignore) initially has a chance higher than $17^k / (2L+1)^{\lfloor k/2 \rfloor}$ to be on the spine, which upper-bounds $\overline{P}_0$. Applying Markov's inequality to $\overline{P}_T - \overline{P}_0$ then upper-bounds $\mathcal{A}$'s success probability. 
\end{proof}

\begin{restatable}{mylemma}{lemmamanyregionqueriesrequired}
\label{lem:many-region-queries-required}

Let $\mathcal{A}$ be a deterministic algorithm for $TARSKI(n, k)$ that receives an input drawn from $\mathcal{U}$. Let $\delta, \epsilon \in (0, 1)$ be such that $2\delta + 2\epsilon < 1$. Suppose $k(n-1)/\rho > 2^{3/(1-2\delta-2\epsilon)}$.

There exists $c = c(\delta, \epsilon)$ such that if $\Pr [\text{$\mathcal{A}$ succeeds}] \geq 1-\delta$, then $\mathcal{A}$ queries spine vertices in at least $c \log(k(n-1)/\rho)$ regions with probability at least $\epsilon$.
\end{restatable}

%

\begin{proof}[Proof sketch] We provide a proof sketch here; the complete proof can be found in Appendix~\ref{app:proof_main_theorem}.

By the construction of a multi-dimensional herringbone function, finding the index $j$ of the fixed point $s^j$ appears to require solving ordered search on the other spine vertices. We formalize this by constructing a randomized algorithm $\mathcal{A}'$ for ordered search that works by simulating $\mathcal{A}$ on an input with a synthetic, random spine. When $\mathcal{A}'$ is itself run on a uniform distribution over ordered search instances, the simulated $\mathcal{A}$ is run on inputs from $\mathcal{U}$. We then use lower bounds on ordered search to complete the proof.
\end{proof}

We can now prove a lower bound on the randomized query complexity of $TARSKI(n,k)$.

\begin{proof}[Proof of Theorem~\ref{thm:lb_k_log_squared_n}]
We proceed by invoking Yao's lemma.
    Let  $\mathcal{U}$ be the uniform distribution over the set of functions $\mathcal{F}$. 
  Let $\mathcal{A}$ be the deterministic algorithm with the smallest possible expected number of queries that succeeds with probability at least $4/5$, where both the expected query count and the success probability are for input drawn from $\mathcal{U}$.
    The algorithm $\mathcal{A}$ exists since there is a finite number of deterministic algorithms for this problem\footnote{Strictly speaking, there are infinitely many algorithms if we allow querying the same vertex multiple times. However, these algorithms are strictly worse than equivalent versions that query each vertex at most once, and there are only finitely many of those.}.
    
    Let $D$ be the expected number of queries issued by $\mathcal{A}$ on input drawn from $\mathcal{U}$.
    Let $R$ be the randomized query complexity of $TARSKI(n,k)$; i.e. the expected number of queries required to succeed with probability at least $9/10$.
    Then Yao's lemma (\cite{yao1977minimaxprinciple}, Theorem 3) yields $2R \ge D$. 
    Therefore it  suffices to lower bound $D$. 

    We argue next that up to a factor of $2$, we can assume that all queries are inside the tube $T_L$. 
    Let  $\overline{TARSKI}(n,k)$ be the Tarski search problem where the input is drawn uniformly at random from $\mathcal{F}$ and    
    the algorithm is only allowed to  query vertices from the tube.
By Lemma \ref{lem:outside-query-simulation}, whenever $\mathcal{A}$ queries a vertex $v$ outside the tube, it could find instead two vertices $a,b$ inside the tube such that the answers to queries $a$ and $b$ carry at least as much information as the answer to query $v$.
Therefore, we can transform $\mathcal{A}$ into an algorithm $\overline{\mathcal{A}}$ for $\overline{TARSKI}(n, k)$ that makes at most $2D$ queries in expectation.

Thus from now on we lower bound the expected number of queries of algorithms that are only allowed to query vertices in the tube.

The success probability of $\overline{\mathcal{A}}$ on inputs drawn from $\mathcal{U}$ is at least $4/5$. By Lemma \ref{lem:many-region-queries-required}, for $\epsilon = 1/5$, there exists $c = c(\epsilon)$ such that, with probability at least $\epsilon$, algorithm $\overline{\mathcal{A}}$ queries spine vertices in at least $c \log (k(n-1)/\rho)$ regions. Whenever it does so, taking every fifth region that $\overline{\mathcal{A}}$ queries a spine vertex in gives at least $m = \left \lfloor \frac{c}{5} \log(k(n-1)/\rho)\right\rfloor$ regions, all of which are distance at least $5$ from each other; in other words, algorithm  $\overline{\mathcal{A}}$ has $m$-surveyed the spine. Therefore:
\begin{align}
    \Pr_{f \sim \mathcal{U}} \left[\text{$\overline{\mathcal{A}}$ $m$-surveys the spine of $f$} \right] &\geq \epsilon \,. 
    \label{eq:abar-likely-to-survey-spine}
\end{align}
%

%
\vspace{-2mm}
By Lemma \ref{lem:many-far-apart-spine-vertices-hard} applied to algorithm $\overline{\mathcal{A}}$, for any $T \in \mathbb{N}$:
\begin{align}
    \Pr_{f \sim \mathcal{U}} \left[ \text{$\overline{\mathcal{A}}$ $m$-surveys the spine within $T$ queries on input $f$}\right] &\leq \frac{15 T \log_2(k)}{\left(m -2\right) \bigl(\lfloor k/2 \rfloor \log_2 (2L + 1) - k \log_2{17}\bigr)}  \label{eq:probability-a-bar-surveys-in-T-small}
\end{align}
Applying  \eqref{eq:probability-a-bar-surveys-in-T-small} with 
\[
T_{\epsilon} = \frac{1}{30 \log_2(k)} \cdot  {\epsilon \left(m-2\right) \cdot \left(\lfloor k/2 \rfloor \log_2 (2L + 1) - k \log_2{17}\right)}\]
gives:
\begin{align}
    \Pr_{f \sim \mathcal{U}} \left[ \text{$\overline{\mathcal{A}}$ $m$-surveys the spine within $T_{\epsilon}$ queries on input $f$}\right] &\leq \frac{\epsilon}{2} \,. 
\end{align}
By \eqref{eq:abar-likely-to-survey-spine}, we must therefore have 
\[ \Pr_{f \sim \mathcal{U}} \left[ \text{$\overline{\mathcal{A}}$ makes more than $T_{\epsilon}$ queries on input $f$}\right] \geq {\epsilon}/{2} \,.
\]

Thus  the expected number $D$ of queries $\overline{\mathcal{A}}$ makes on inputs drawn from $\mathcal{U}$ must be at least:
\begin{align}
    D &\geq T_{\epsilon} \cdot \frac{\epsilon}{2} 
    = \frac{\epsilon^2 \left(\left\lfloor \frac{c}{5} \log(k(n-1) / \rho) \right\rfloor-2\right) \cdot \left(\lfloor k/2 \rfloor \log_2 (2L + 1) - k \log_2{17}\right)}{60 \log_2(k)} \,. \label{eq:D_lower_bound_final}
\end{align}
Since $\epsilon$ and $c$ are constants, setting $\rho = k\sqrt{n-1}$ and $ L = \frac{1}{12} \sqrt{n-1}$ in  \cref{eq:D_lower_bound_final} implies that  $D \in \Omega(k \log^2(n) / \log(k))$, as desired.
\end{proof}



\section{An $O(k \log n \log(nk))$ upper bound for multi-dimensional  herringbone functions} \label{sec:upper_bound_distribution_U}

We also provide a deterministic upper bound on the query complexity of finding a fixed point of any multi-dimensional herringbone function, which is quite close to the lower bound of \cref{thm:lb_k_log_squared_n}. We include the theorem statement and proof sketch here; see Appendix~\ref{app:upper_bound_distribution_U} for the full proof.

\begin{restatable}{mytheorem}{theoremherringboneupperbound}
 \label{thm:herringbone-upper-bound}
    There exists an $O(k \log(n) \log(nk))$-query algorithm to find the fixed point of a multi-dimensional herringbone function.
\end{restatable}
\begin{proof}[Proof sketch]
Our algorithm uses an $O(k \log n)$-query subroutine that finds, for an input $x$, the spine vertex with Hamming weight $x$. This subroutine exploits the dependence of the off-spine function values on the shape of the spine to perform $k$ independent binary searches along the $k$ axes. Using this subroutine, we can run binary search along the spine to find the fixed point, which takes $O(\log (nk))$ iterations.
\end{proof}

\bibliographystyle{alpha}

\bibliography{tarski_bib}

\appendix 

\section{Multi-dimensional Herringbone Functions} \label{app:properties_herringbone_functions}

In this section, we include several proofs showing  properties of multi-dimensional herringbone functions.

\lemmauniquefixedpoint*
\begin{proof}
    By definition, $s^j$ is a fixed point. No other point on the spine can be a fixed point, since they each map to a different point on the spine.

    All that remains is to show that there are no fixed points off the spine. Let $v$ be an arbitrary vertex not on the spine. By definition, $h^{\mathbf{s}, j}(v)$ differs from $v$ in two ways:
    \begin{itemize}
        \item The coordinate $a$ for which $s^{\mu(v)+1}_a = s^{\mu(v)}_a + 1$ is increased by $1$, and
        \item The coordinate $b$ for which $s^{M(v)}_b = s^{M(v)-1}_b + 1$ is decreased by $1$.
    \end{itemize}
    As long as $a \neq b$, these two changes do not cancel each other out, and $v$ is not a fixed point.
    
    Assume for contradiction that $a = b$. Then, since $s^{\mu(v)} \leq v$ but $s^{\mu(v) + 1} \not \leq v$, we have $v_a < s^{\mu(v)+1}_a$. Similarly, we have $s^{M(v) - 1}_b < v_b$, or equivalently $s^{M(v) - 1}_a < v_a$. Since $v$ is not on the spine, moving along the spine from a point $\leq v$ to a point $\geq v$ requires at least two moves, so $\mu(v) + 2 \leq M(v)$; therefore, $s^{\mu(v)+1} \leq s^{M(v)-1}$, so we have:
    \begin{equation}
        v_a < s^{\mu(v)+1}_a \leq s^{M(v)-1}_a < v_a,
    \end{equation}
    which is a contradiction. Therefore, $a \neq b$, so $v$ is not a fixed point.
\end{proof}

Several later arguments will be simplified by the following symmetry.
\begin{lemma}\label{lem:herringbone-axis-inversion-symmetry}
    Let $\vec{s}$ be a spine and let $j \in \{0, \ldots, (n-1)k\}$ be a fixed point index. Let $\vec{t}$ be the spine created by inverting $\vec{s}$ along all axes, so that for all $i \in \{0, \ldots, (n-1)k\}$:
    \begin{align}
        \vec{t}^i = (n-1, \ldots, n-1) - \vec{s}^{(n-1)k-i}
    \end{align}
    Then the corresponding multi-dimensional herringbone function $h^{\vec{t}, (n-1)k-j}$ is an axis-inverted version of $h^{\vec{s}, j}$. Specifically, for all $v \in \{0, \ldots, (n-1)\}^k$, we have:
    \begin{align}\label{eq:axis-inversion-symmetry-condition}
        h^{\vec{t}, (n-1)k-j}(v) = (n-1, \ldots, n-1) - h^{\vec{s}, j}\bigl( (n-1, \ldots, n-1) - v \bigr)
    \end{align}
\end{lemma}

\begin{proof}
    Consider an arbitrary $v \in \{0, \ldots, (n-1)\}^k$. The first three cases of Definition \ref{def:general-herringbone} are straightforward:
    \begin{itemize}
        \item If $v = t^{(n-1)k-j}$, then by definition $(n-1, \ldots, n-1) - v = s^j$. Therefore, both inputs are fixed points of their respective functions, so \eqref{eq:axis-inversion-symmetry-condition} simplifies to $v=v$.
        \item If $v = t^i$ for some $i < (n-1)k-j$, then $h^{\vec{t}, (n-1)k-j}(v) = t^{i+1}$. We would also have $(n-1, \ldots, n-1) - v = s^{(n-1)k-i}$, where $(n-1)k-i > j$; therefore, $h^{\vec{s}, j}\bigl( (n-1, \ldots, n-1) - v \bigr) = s^{(n-1)k - i - 1} = (n-1, \ldots, n-1) - t^{i+1}$. Plugging these values into \eqref{eq:axis-inversion-symmetry-condition} satisfies the equation.
        \item If $v = t^i$ for some $i > (n-1)k-j$, then $h^{\vec{t}, (n-1)k-j}(v) = t^{i-1}$. We would also have $(n-1, \ldots, n-1) - v = s^{(n-1)k-i}$, where $(n-1)k-i < j$; therefore, $h^{\vec{s}, j}\bigl( (n-1, \ldots, n-1) - v \bigr) = s^{(n-1)k - i + 1} = (n-1, \ldots, n-1) - t^{i-1}$. Plugging these values into \eqref{eq:axis-inversion-symmetry-condition} satisfies the equation.
    \end{itemize}
    For the general case that $v$ is not on the spine $\vec{t}$, we need to consider $M(v)$ and $\mu(v)$. To disambiguate these functions for the two different spines, let $M_s$ and $\mu_s$ denote the versions for $\vec{s}$; let $M_t$ and $\mu_t$ denote the versions for $\vec{t}$.

    We first consider $M_t(v)$:
    \begin{align}
        M_t(v) &= \min_{t^i \geq v} i \\
        &= \min_{s^{(n-1)k-i} \leq (n-1, \ldots, n-1) - v} i \\
        &= (n-1)k - \max_{s^{(n-1)k - i} \leq (n-1, \ldots, n-1) - v} (n-1)k - i \\
        &= (n-1)k - \mu_s\bigl((n-1, \ldots, n-1) - v\bigr)
    \end{align}
    By similar algebra, for $\mu_t(v)$ we have:
    \begin{align}
        \mu_t(v) &= \max_{t^i \leq v} i \\
        &= \min_{s^{(n-1)k-i} \geq (n-1, \ldots, n-1) - v} i \\
        &= (n-1)k - \min_{s^{(n-1)k - i} \geq (n-1, \ldots, n-1) - v} (n-1)k - i \\
        &= (n-1)k - M_s\bigl((n-1, \ldots, n-1) - v\bigr)
    \end{align}
    We can now compute the two ``corrections" in $h^{\vec{t}, (n-1)k-j}(v)$ from $v$. First, the one involving $M_t(v)$:
    \begin{align}
        t^{M_t(v)} - t^{M_t(v)-1} &= t^{(n-1)k - \mu_s\bigl((n-1, \ldots, n-1) - v\bigr)} - t^{(n-1)k - \mu_s\bigl((n-1, \ldots, n-1) - v\bigr) - 1} \\
        &= s^{\mu_s\bigl((n-1, \ldots, n-1) - v\bigr)} - s^{\mu_s\bigl((n-1, \ldots, n-1) - v\bigr) + 1} \label{eq:t-M-correction-equals-s-mu-correction}
    \end{align}
    For $\mu_t(v)$ we have:
    \begin{align}
        t^{\mu_t(v) + 1} - t^{\mu_t(v)} &= t^{(n-1)k - M_s\bigl((n-1, \ldots, n-1) - v\bigr) + 1} - t^{(n-1)k - M_s\bigl((n-1, \ldots, n-1) - v\bigr)} \\
        &= s^{M_s\bigl((n-1, \ldots, n-1) - v\bigr) - 1} - s^{M_s\bigl((n-1, \ldots, n-1) - v\bigr)} \label{eq:t-mu-correction-equals-s-M-correction}
    \end{align}
    Using \eqref{eq:t-M-correction-equals-s-mu-correction} and \eqref{eq:t-mu-correction-equals-s-M-correction}, we can compute $h^{\vec{t}, (n-1)k-j}(v)$:
    \begin{align}
        &h^{\vec{t}, (n-1)k-j}(v) \\
        = &v + \left(t^{\mu_t(v) + 1} - t^{\mu_t(v)}\right) - \left(t^{M_t(v)} - t^{M_t(v)-1}\right) \\
        =& v + \left(s^{M_s\bigl((n-1, \ldots, n-1) - v\bigr) - 1} - s^{M_s\bigl((n-1, \ldots, n-1) - v\bigr)}\right) - \left(s^{\mu_s\bigl((n-1, \ldots, n-1) - v\bigr)} - s^{\mu_s\bigl((n-1, \ldots, n-1) - v\bigr) + 1}\right) \\
        =& (n-1, \ldots, n-1) - h^{\vec{s}, j}\bigl( (n-1, \ldots, n-1) - v\bigr),
    \end{align}
    which completes this final case.
\end{proof}

\lemmamonotone*
\begin{proof}
    Consider two vertices $v$ and $w$ such that $v \leq w$. There are several cases, depending on whether $v$, $w$, both, or neither are on the spine.
    \paragraph{Case $v$, $w$ both on the spine.} Then $h^{\mathbf{s}, j}(v)$ and $h^{\mathbf{s}, j}(w)$ are both on the spine as well. They can't pass each other since then they'd pass the fixed point, but they only go one step along the spine at a time. Therefore, $h^{\mathbf{s}, j}(v) \leq h^{\mathbf{s}, j}(w)$.
    \paragraph{Case $v$ not on spine, $w$ on spine.} We first show that $h^{\mathbf{s}, j}(v) \leq s^{M(v)-1}$. Since $v \leq s^{M(v)}$, the coordinates $i$ for which $(h^{\mathbf{s}, j}(v))_i \leq v_i$ and $s^{M(v)}_i = s^{M(v)-1}_i$ are safe. The only possible exceptions are $a$ and $b$, where $(h^{\mathbf{s}, j}(v))_a = v_a + 1$ (and therefore $s^{\mu(v)+1}_a = s^{\mu(v)}_a+1$) and $s^{M(v)}_b = s^{M(v)-1}_b + 1$. These are precisely the coordinates where $v$ differs from $h^{\mathbf{s}, j}(v)$, so by Lemma \ref{lem:h-unique-fixedpoint}, $a \neq b$.
    \begin{itemize}
        \item[$a$.] Since $s^{\mu(v)} \leq v$ but $s^{\mu(v)+1} \not \leq v$, we must have $v_a + 1 = s^{\mu(v)+1}_a \leq s^{M(v)-1}_a$. Therefore, $(h^{\mathbf{s}, j}(v))_a = v_a + 1 \leq s^{M(v)-1}_a$.
        \item[$b$.] By definition, $(h^{\mathbf{s}, j}(v))_b = v_b - 1$. Since $s^{M(v)} \geq v$ but $s^{M(v)-1} \not \geq v$, we must have $v_b - 1 = s^{M(v)-1}_b$. Therefore, $(h^{\mathbf{s}, j}(v))_b = s^{M(v)-1}_b$.
    \end{itemize}
    Therefore, $h^{\mathbf{s}, j}(v) \leq s^{M(v)-1} \leq h^{\mathbf{s}, j}(w)$.
    \paragraph{Case $v$ on spine, $w$ not on spine.} By applying Lemma \ref{lem:herringbone-axis-inversion-symmetry}, we can exchange the roles of $v$ and $w$, making this case symmetric to the previous.
    \paragraph{Case $v$, $w$ both not on the spine.} Only two coordinates are possibly a problem: the one where $v$ increases, and the one where $w$ decreases. By Lemma \ref{lem:herringbone-axis-inversion-symmetry}, inverting the axes maps one of these coordinates into the other, so we only consider the coordinate $a$ on which $v$ increases. There are three subcases depending on how $v_a$ compares to $w_a$.
    \begin{itemize}
        \item \textbf{Case $v_a + 2 \leq w_a$.} We have $(h^{\mathbf{s}, j}(v))_a = v_a + 1 \leq w_a - 1 \leq (h^{\mathbf{s}, j}(w))_a$.
        \item \textbf{Case $v_a + 1 = w_a$.} This is only a problem if $w$ decreases in the $a$ dimension. Assume for contradiction that it does. Then $s^{M(w)-1}_a = s^{M(w)}_a - 1$. Since $s^{M(w)} \geq w$ but $s^{M(w)-1} \not \geq w$, we have $w_a - 1 = s^{M(w)}_a - 1 = s^{M(w)-1}_a$. Similarly, $v_a + 1 = s^{\mu(v)+1}_a$. But $\mu(v) + 1 \leq M(v) - 1 \leq M(w) - 1$, so $v_a + 1 \leq w_a - 1$, contradicting $v_a + 1 = w_a$. Therefore, $w$ does not decrease in the $a$ dimension, so $(h^{\mathbf{s}, j}(v))_a \leq (h^{\mathbf{s}, j}(w))_a$.
        \item \textbf{Case $v_a = w_a$.} As before, $v_a = s^{\mu(v)}_a = s^{\mu(v) + 1}_a - 1$. But since $w_a = v_a$, we then have $w \not \geq s^{\mu(v)+1}$. Since $w \geq v \geq \mu(v)$, we have $\mu(w) = \mu(v)$, so $w$ also increases in the $a$ dimension.     
        \end{itemize}

    In all cases, $h^{\mathbf{s}, j}(v) \leq h^{\mathbf{s}, j}(w)$, so $h^{\mathbf{s}, j}$ is monotone.
\end{proof}

\section{Hard Distribution} \label{app:hard_distribution}

Here we include the proofs of statements  from Section~\ref{sec:hard_distribution}.

\lemmaoutsidequerysimulation*
\begin{proof}
    The desired $a$ and $b$ can be defined coordinate-wise as:
    \begin{align}
        a_i &= \min \{ \min_j v_j + 2L, v_i\} \\
        b_i &= \max \{ \max_j v_j - 2L, v_i\}
    \end{align}
    
    Computing $f(v)$ requires knowing the increasing component $s^{\mu(v)+1} - s^{\mu(v)}$ and the decreasing component $s^{M(v)} - s^{M(v)-1}$. We will show that $\mu(v) = \mu(a)$, which is sufficient for the first of these as $s^{\mu(a)+1} - s^{\mu(a)}$ is the increasing component of $f(a) - a$. This will also cover $b$, as inverting the axes swaps the roles of $a$ and $b$ and, by Lemma \ref{lem:herringbone-axis-inversion-symmetry}, will mean $f(b)$ identifies the decreasing component of $f(v)$.

    For a point $x$, let $D(x) = \{ y \mid y \leq x\}$ be the set of points dominated by $x$ in $\{0, \ldots, n-1\}^k$. It will be shown that:
    \begin{equation}
        T_L \cap D(v) = T_L \cap D(a)
    \end{equation}

    Certainly $T_L \cap D(v) \supseteq T_L \cap D(a)$, since $v \geq a$. To show the other direction, rewrite $T_L$ as the union of $k$ hypercubes:
    \begin{align}
        C_L^i &= \{ y \mid \text{for all $j \in [k]$, $|y_j - i| \leq L$}\} \\
        T_L &= \bigcup_{i=1}^k C_L^i
    \end{align}

    For $i > \min_j v_j + L$, $C_L^i$ is disjoint from $D(v)$, as the minimum coordinate of $v$ is then too small. For $i \leq \min_j v_j + L$, they are not disjoint, and in fact we claim each such intersection has a unique maximum $c^i$ whose coordinates are:
    \begin{equation}
        c^i_{j} = \min \{i + L, v_j\}
    \end{equation}

    Any point not dominated by $c^i$ will either fail to be in $C_L^i$ (by having a coordinate larger than $i+L$) or fail to be in $D(v)$ (by having a coordinate larger than one of $v$'s). Each $c^i$ lies in the corresponding $C_L^i$, as $|(i+L) - i| = L$ and, for $v_j \leq i+L$:
    \begin{align}
        -L &\leq i - v_j \\
        &\leq (\min_{\ell} v_{\ell} + L) - v_j \\
        &\leq L
    \end{align}

    Each $c^i$ is also in $D(v)$, as $c^i \leq v$. Therefore, for $i \leq \min_j v_j + L$, the unique maximum element of $C_L^i$ is $c^i$. Since the formula for $c^i$ is monotone increasing in $i$, the maximum of these maxima is $c^{\min_j v_j + L} = a$. Therefore, $T_L \cap D(v) \subseteq T_L \cap D(a)$, which together with the other direction shows $T_L \cap D(v) = T_L \cap D(a)$. Since $\mu(v)$ is the index of the maximum spine vertex within $D(v)$ and the spine exists entirely within $T_L$, $\mu(v) = \mu(a)$.

\end{proof}

\lemmatubequerylocality*
\begin{proof}
    If $v$ is on the spine, then $f(v)$ is by definition determined by the spine vertex $v$ and its neighbors.

    If $v$ is not on the spine, the value of $f(v)$ depends on $\mu(v)$ and $M(v)$. The spine vertices $s^{\mu(v)}$ and $s^{M(v)}$ each share a coordinate with $v$, so bounding the distance between points in $T_L$ which share a coordinate will bound the locality of $f$.

    Fix a coordinate index $\kappa \in [k]$ and coordinate value $r \in \{0, \ldots, n-1\}$. Rewrite $T_L$ as the union of hypercubes:
    \begin{align}
        C_L^i &= \{ y \mid \text{for all $j \in [k]$, $|y_j - i| \leq L$}\} \\
        T_L &= \bigcup_{i=0}^{n-1} C_L^i
    \end{align}
    For $i \not \in [r - L, r + L]$, $C_L^i$ contains no points with $\kappa$-coordinate $r$. For $i \in [r - L, r + L]$, all points in $C_L^i$ have values in $[i - L, i + L] \subseteq [r - 2L, r + 2L]$. Therefore, the coordinates of $v$ can only differ from those of $s^{\mu(v)}$ and $s^{M(v)}$ by $2L$. The value of $f(v)$ depends on $s^{\mu(v)}$, $s^{M(v)}$, and the closer vertices $s^{\mu(v) + 1}$ and $s^{M(v) - 1}$. All of these have coordinates within $2L$ of those of $v$, which completes the proof.
\end{proof}

\lemhelpdefinitionspine*
\begin{proof}
    We first show that $z$ is well-defined. Since $\mathcal{C}(v)$ is closed and convex, the set $\mathcal{C}(v) \cap \ell(a, b)$ is itself a closed line segment. Because $a \leq b$, the slope of $\ell(a, b)$ is nonnegative in every coordinate. Therefore, the higher endpoint of $\mathcal{C}(v) \cap \ell(a, b)$ has coordinates which are weakly larger than every other point in $\mathcal{C}(v) \cap \ell(a, b)$, so it can be taken as $z$.

    If $v=b$, then we are done. Otherwise, suppose $v \neq b$.
    For contradiction, suppose $z$ did not lie on any forward face of $\mathcal{C}(v)$. Then there exists some $\varepsilon > 0$ such that, if each coordinate of $z$ is increased by an amount in $[0, \varepsilon]$, the resulting point remains in $\mathcal{C}(v)$. 
    Since $a \leq b$, moving along $\ell(a, b)$ can only increase coordinates. 
Let $x$ be  a point on the line $\ell(a, b)$ such that $x > z$ and $ \|x-z\|_{\infty} \leq \varepsilon$. Such a point $x$ exists since $z < b$ due to the requirement $v \neq b$. Moreover, by the choice of $\varepsilon$, we have $x \in \mathcal{C}(v)$. This contradicts $z$ being the maximum point in $\mathcal{C}(v) \cap \ell(a, b)$.
    Thus the assumption must have been false and  $z$ lies on some forward face of $\mathcal{C}(v)$, as required.
\end{proof}

\lemsequencesjbetweenchis*
\begin{proof}
We show by induction on $i$ that a sequence of vertices $s^1, s^2, \ldots, s^i$ with $s^1 = u$ can be chosen such that 
\begin{enumerate}[(I)]
\item $(s^{1}, \ldots, s^{i})$ is a sequence of vertices in the lattice graph $\{0, \ldots, n-1\}^k$ such that for some $m \leq i$, we have  $(s^1, \ldots, s^m)$ is a monotone connected path  and either $s^m = v$ or $m=i$. 
\item $\mathcal{C}(s^i) \cap \ell(u, v) \neq \emptyset$.
\end{enumerate}

 The base case is $i=1$, using $m=i=1$ and $s^1 = u$. The path $(s^1)$ consists of one point, so it is clearly monotone. Since one endpoint of $\ell(u, v)$ is $u$, the cube  $\mathcal{C}(u)$ certainly intersects $\ell(u, v)$. 

\begin{figure}[h!]
\centering 
\includegraphics[scale=0.368]{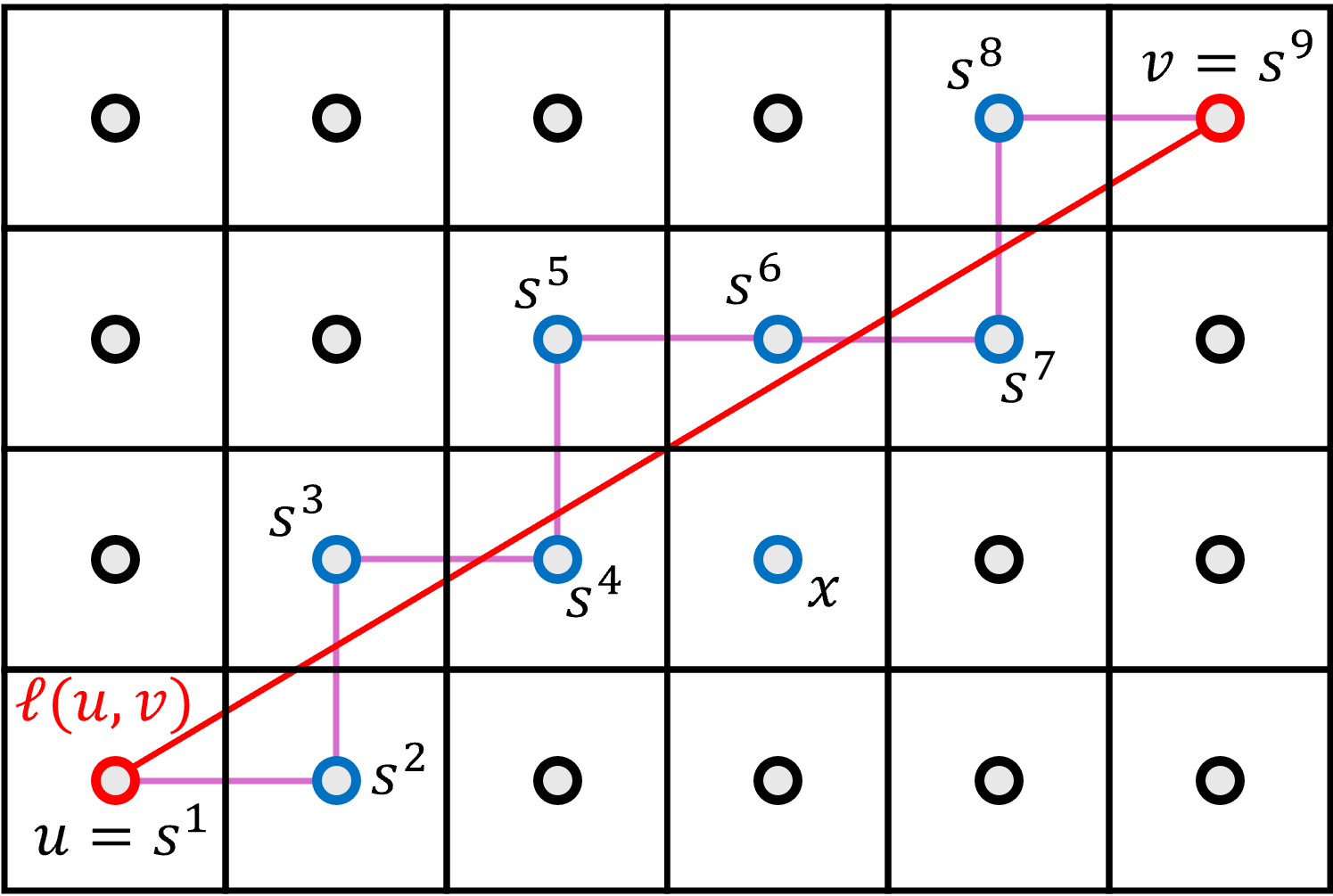}
\caption{A 2-dimensional example of the construction of Lemma \ref{def:sequence_s_j_between_chi_i_and_chi_i_plus_rho}, with endpoints $u$ and $v$ and the associated spine vertices $s^1$ through $s^9$. Each vertex is drawn surrounded by its cube. The vertex $x$ was also an option for $s^5$ due to $\ell(u, v)$ passing through a corner of the cubes, but was arbitrarily discarded.}
\label{fig:line-through-grid}
\end{figure}

     Assume the sequence of vertices $s^1, \ldots, s^{i-1}$ has been chosen such that (I) and (II) hold for some $m \leq i-1$. We consider three cases:
     \begin{description} 
     \item[$(1)$] If $m < i-1$, then we are guaranteed that $s^m = v$. Setting $s^i = s^m$ will  satisfy (I) and (II).
     \item[$(2)$] If $m = i-1$ and $s^{i-1} = v$, then the construction is finished: we can take $m = i - 1$ and $s^i = s^m$.
     \item[$(3)$] Else,
     let 
     $ z_{i-1} = \max \{\mathcal{C}({s^{i-1}}) \cap \ell(u, v)\}\,.$ 
     The point $z_{i-1}$ is  well-defined by part (II) of the inductive hypothesis.    

    For each coordinate $j \in [k]$, let $w_j = s^{i-1} + \vec{e}_j$. The  face of $\mathcal{C}(s^{i-1})$ with maximum value in coordinate $j$ is shared with the cube $\mathcal{C}(w_j)$. By Lemma~\ref{lem:help_definition_spine}, the point $z_{i-1}$ lies on a forward face of $\mathcal{C}({s^{i-1}})$, so there exists $j \in [k]$ such that  $z_{i-1} \in \mathcal{C}(w_j)$. Then choose $s^{i}$ to be the lexicographically smallest point in the set $\{w_1, \ldots, w_k\}$ with the property that $z_{i-1} \in \mathcal{C}(s^{i})$. 

\begin{figure}[h!]
\centering 
\includegraphics[scale=0.69]{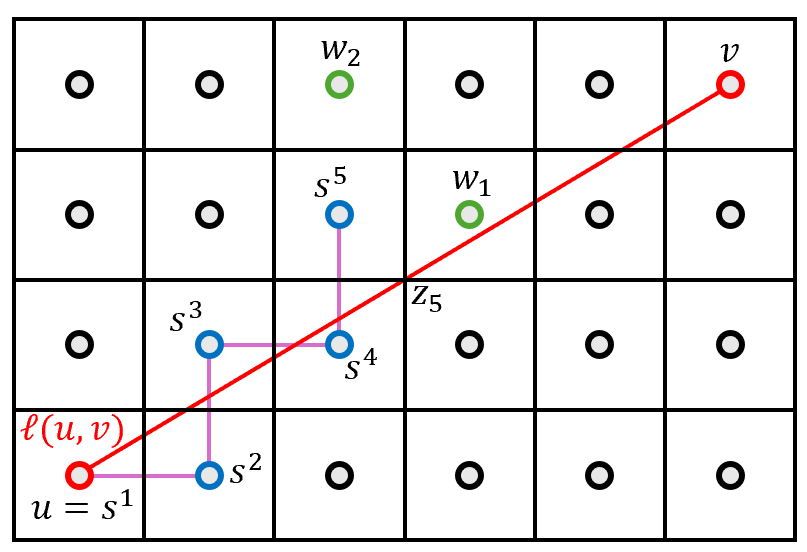}
\caption{
The example from Figure \ref{fig:line-through-grid}, where so far only $\vec{s} = (s^1, \ldots, s^5)$ has been selected. In the next iteration ($i=6$),  we need to choose $s^6$. Since $v \not \in \vec{s}$, it must be the case that $m =5$. Then we are in case (3), which defines a  point $z_5 = \max \{\mathcal{C}({s^{5}}) \cap \ell(u, v)\}$, and two candidates $w_1$ and $w_2$ for $s^6$. Of these two candidates, only $w_1$ has the property that $z_5 \in \mathcal{C}(w_1)$, so 
we must choose $s^6 = w_1$. }
\end{figure}
    
 We  show that (I) and (II) continue to hold. Since $(s^1, \ldots, s^{i-1})$ is a monotone connected path and $s^i = s^{i-1} + \vec{e}_j$ for some $j \in [k]$, the whole path $(s^1, \ldots, s^i)$ is also monotone and connected. Thus (I) holds. Also, $z_{i-1} \in \mathcal{C}(s^i)$, so (II) holds, completing the induction step.
 \end{description}

What remains  is showing that this construction eventually triggers the $s^m = v$ case. The line segment $\ell(u, v)$ is finite, so it can be enclosed by some finite hypercube. However, because $s^1, s^2, \ldots$ is a strictly increasing monotone sequence in $\{0, \ldots, n-1\}^k$, it will eventually leave any bounded region. Since each $s^i$ satisfies $\mathcal{C}(s^i) \cap \ell(u, v) \neq \emptyset$, there can only be finitely many $s^i$ before $s^m = v$.
\end{proof}

\lemqueryaffectsadjacentonly*
\begin{proof}
     By Lemma \ref{lem:tube-query-locality}, the value of $f(v)$ depends only on spine vertices with coordinates each within $2L$ of $v$. In particular, this only includes vertices with Hamming weight within $2kL$ of $wt(v)$. Because $\rho > 2kL$, all such vertices are in one of $R_{\alpha - \rho}$, $R_{\alpha}$, and $R_{\alpha + \rho}$ -- up to one region away from $v$.

    These spine vertices in turn are defined by the spine's endpoints in the boundaries of $R_{\alpha - \rho}$, $R_{\alpha}$, and $R_{\alpha + \rho}$; specifically, the spine's intersection with $Low_{\alpha + i\rho}$, for $i \in \{-1, \ldots, 2\}$. By the construction of $\mathcal{U}$, the spine's intersections with all further boundaries are independent of these. Therefore, the spine in regions $R_{\alpha - \rho}$ through $R_{\alpha + \rho}$ is independent of the spine in any region $R_{\alpha + i\rho}$ for $i < -2$ or $i > 2$, which completes the proof.
\end{proof}

\section{Proof of the Main Theorem} \label{app:proof_main_theorem}

This section includes the proofs of lemmas from Section~\ref{sec:proof_main_theorem}.

\lemmasizelowa*
\begin{proof}
    For all $x_1, \ldots, x_{\lfloor k/2 \rfloor} \in \{-L, \ldots, L\}$, define 
    \begin{align}
        g(x_1, \ldots, x_{\lfloor k/2 \rfloor}) = \begin{cases}
            \left( \frac{a}{k} + x_1, \frac{a}{k} - x_1, \frac{a}{k} + x_2, \frac{a}{k} - x_2, \ldots, \frac{a}{k} + x_{\lfloor \frac{k}{2} \rfloor}, \frac{a}{k} - x_{\lfloor \frac{k}{2} \rfloor}\right) & \text{if $k$ is even} \\
            \left( \frac{a}{k} + x_1, \frac{a}{k} - x_1, \frac{a}{k} + x_2, \frac{a}{k} - x_2, \ldots, \frac{a}{k} + x_{\lfloor \frac{k}{2} \rfloor}, \frac{a}{k} - x_{\lfloor \frac{k}{2}\rfloor}, \frac{a}{k}\right) & \text{if $k$ is odd}
        \end{cases}
    \end{align}
    We show that $g(x_1, \ldots, x_{\lfloor \frac{k}{2}\rfloor}) \in Low_a$ by arguing that  $g(x_1, \ldots, x_{\lfloor \frac{k}{2}\rfloor}) \in T_L$ and its Hamming weight is $a$. 

    We first argue that $g(x_1, \ldots, x_{\lfloor \frac{k}{2} \rfloor}) \in \{0, \ldots, n-1\}^k$. Let $i = a/k$. Since $a \in \mathcal{I}$, we have $k \mid a$, so $i \in \mathbb{N}$ and $g(x_1, \ldots, x_{\lfloor \frac{k}{2} \rfloor}) \in \mathbb{Z}^k$. Because $a \in \mathcal{I} \setminus \{0, k(n-1)\}$, by definition of $\mathcal{I}$ it must be the case that $a \in [\rho, k(n-1)-\rho]$, so:
    \begin{align}
        \frac{a}{k} \geq \frac{\rho}{k} \geq 12L 
        \qquad \mbox{and} \qquad
        \frac{a}{k} \leq \frac{k(n-1)-\rho}{k} \leq (n-1) - 12L \,.
    \end{align}
    Therefore adding terms of magnitude at most $L$ to each coordinate of $(i, \ldots, i)$ leaves $g(x_1, \ldots, x_{\lfloor \frac{k}{2} \rfloor})$ in $\{0, \ldots, n-1\}^k$.
    
    We now argue that $g(x_1, \ldots, x_{\lfloor \frac{k}{2} \rfloor}) \in T_L$ because each of its components is within $L$ of $i$. For each $j \in \{1, \ldots, \lfloor k/2 \rfloor \}$, we have 
    \begin{align}
        \left| \frac{a}{k} + x_j - i \right| = |x_j| \leq L \qquad \mbox{and} \qquad  
        \left| \frac{a}{k} - x_j - i \right|  = |x_j| \leq L \,.
    \end{align}
    Thus the condition \eqref{eq:condition_to_be_in_tube} from the definition of the tube holds, so $g(x_1, \ldots, x_{\lfloor \frac{k}{2} \rfloor}) \in T_L$.
    
    We have that it has Hamming weight $a$ 
    because each of the $x_j$'s cancel out when adding up the coordinates. Therefore, $g(x_1, \ldots, x_{\lfloor \frac{k}{2} \rfloor}) \in Low_a$.

As different values of the $x_j$'s result in different values of their corresponding coordinates, each choice of $x_1, \ldots, x_{\lfloor \frac{k}{2} \rfloor} \in \{-L, \ldots, L\}$ yields a distinct point. Thus $|Low_a| \geq (2L+1)^{\lfloor k/2 \rfloor}$ as required.
\end{proof}

\lemmarandomspineunlikelyhitvertex*

\begin{proof} 
    For each $u \in Low_a$, we have $wt(u) = a$. Similarly, for each $v \in Low_b$, we have $wt(v) = b$. Lemma~\ref{def:sequence_s_j_between_chi_i_and_chi_i_plus_rho} gives a  monotone path  $\vec{s}(u,v) = (s^1, \ldots, s^m)$ from $u$ to $v$, so each vertex $s^i$ on the path  satisfies $a \leq wt(s^i) \leq b$. 
    If $wt(w) < a$ or $wt(w) > b$, then $w$ cannot be one of the vertices $s^i$, so $Pr\bigl[w \in \vec{s}(U,V) \bigr] = 0$, thus inequality \eqref{eq:w_hits_sequence_s} holds.

    Thus from now we can assume  $wt(w) \in [a, b]$.
   Without loss of generality, we assume $w$ lies closer to $Low_b$ than $Low_a$, i.e.:
    \begin{align}\label{eq:wt-w-ge-a-plus-b-over-2}
        wt(w) \geq \frac{a+b}{2}\,.
    \end{align}
    Fix an arbitrary vertex ${u} \in Low_a$. Let $B$ denote the hyperplane:
    \begin{align}
        B = \{\vec{x} \in \mathbb{R}^k \mid x_1 + \ldots + x_k = b\}
    \end{align}
   Let $\mathcal{P}(w)$ denote the projection of $\mathcal{C}(w)$ onto $B$ from $u$, i.e.:
    \begin{align}
        \mathcal{P}(w) = \{ \vec{x} \in B \mid \mathcal{C}(w) \cap \ell(u, \vec{x}) \neq \emptyset \},
    \end{align}
  recalling that $\mathcal{C}(w)$ is the cube of side-length $1$ centered at $w$. 
  
  By construction, $\vec{s}(u, V)$ consists only  of vertices $x \in \{0,\ldots, n-1\}^k$ for which $\mathcal{C}(x)$ intersects the line segment $\ell(u,V)$ from $u$ to $V$. Then   $w \in \vec{s}(u, V)$ only if $V \in \mathcal{P}(w)$.

    To bound the number of lattice points in $\mathcal{P}(w)$, we will upper-bound the distance between any two points in $\mathcal{P}(w)$. Consider two arbitrary points $p, q \in \mathcal{P}(w)$. Arbitrarily choose $\overline{p} \in \mathcal{C}(w) \cap \ell(u, p)$ and $\overline{q} \in \mathcal{C}(w) \cap \ell(u, q)$.

    We claim that $\overline{p} \neq u$ and show it by arguing the two points have very different Hamming weight. Consider an arbitrary $x \in \mathcal{C}(w)$. Since $\|x-w\|_{\infty} \leq 1/2$, we have 
\begin{align} \label{eq:lb_wt_x}
wt(x)  &\geq wt(w) - \frac{1}{2} \cdot k \,.
\end{align}
    Then 
    \begin{align} 
        wt(x) - wt(u) &\geq wt(w) - \frac{1}{2} \cdot k - wt(u) \explain{By \eqref{eq:lb_wt_x}}  \\
        &\geq \frac{a+b}{2} - \frac{k}{2} - a = \frac{b-a}{2} - \frac{k}{2}\explain{By \eqref{eq:wt-w-ge-a-plus-b-over-2}} \\
        &\geq \frac{b-a}{2} - \frac{b-a}{24L} \explain{By choice of $a,b$} \notag \\
        &\geq \frac{11}{24}(b-a) \,. \label{eq:big_difference_wt_x_and_u}
    \end{align}
Inequality \eqref{eq:big_difference_wt_x_and_u} implies that $x \neq u$. In particular, we have $wt(\overline{p}) > wt(u)$ since $\overline{p} \in \mathcal{C}(w)$ and $b-a > 0$.

    Then define
    \begin{align}\label{eq:q-tilde-in-terms-of-pbar}
        \widetilde{q} = u + (\overline{q} - u) \frac{wt(p - u)}{wt(\overline{p} - u)}\,.
    \end{align}
    That is, $\widetilde{q}$ is the point that $\overline{q}$ would map to if  $\mathcal{C}(w)$ were proportionally scaled from $u$ by the factor required to map $\overline{p}$ to $p$.

\begin{figure}[h!]
    \centering 
    \includegraphics[scale=0.6]{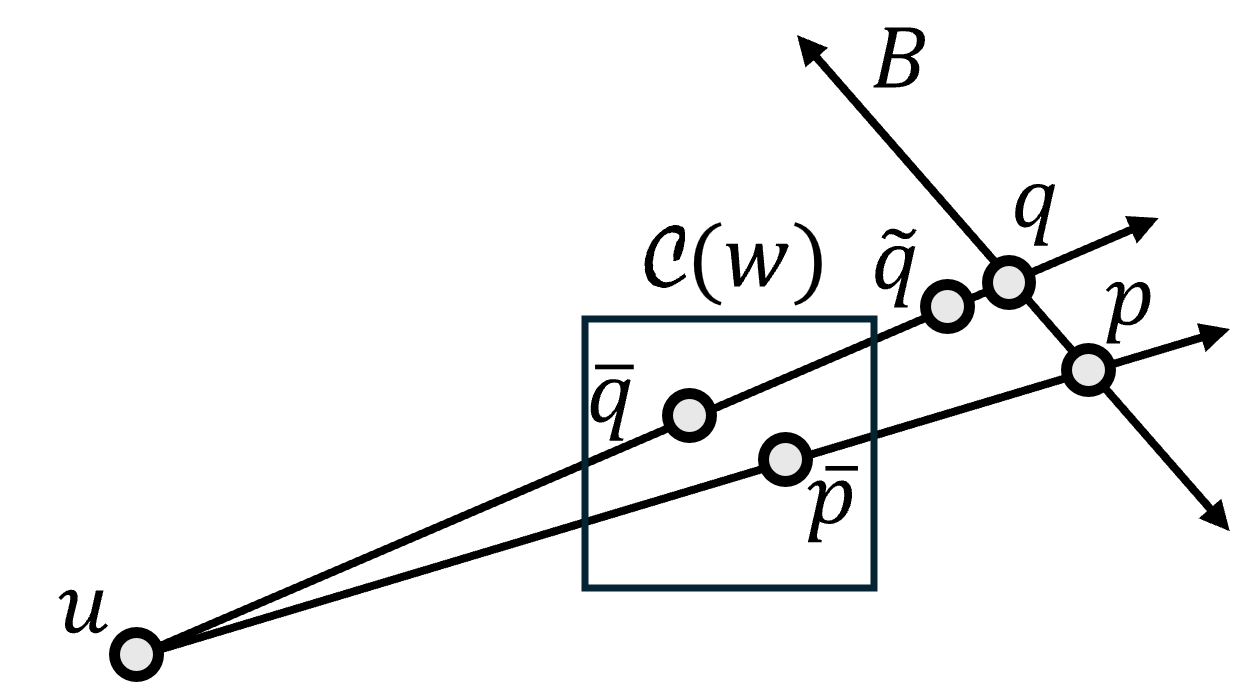}
    \caption{The geometry of Lemma \ref{lem:random-spine-unlikely-to-hit-vertex}, including only what is defined through equation \eqref{eq:q-tilde-in-terms-of-pbar}. The point $u \in Low_a$ is arbitrary. Points $p$ and $q$ lie on the plane $B$ that contains $Low_b$ such that $\mathcal{C}(w)$ (for a vertex $w$ not pictured) lies between them and $u$. Lines from $u$ to each of $p$ and $q$ are also drawn, with points $\overline{p}$ and $\overline{q}$ being arbitrary points within $\mathcal{C}(w)$ and each of these lines. Point $\widetilde{q}$ is chosen so that triangles $\overline{q}u\overline{p}$ and $\widetilde{q}up$ are similar.}
    \label{fig:random-spine-lemma-geometry-lite}
\end{figure}

    Since  $\overline{p}, \overline{q} \in \mathcal{C}(w)$, which is a hypercube of side length $1$, we have:
    \begin{equation}\label{eq:p-q-norm-leq-1}
        \|\overline{p} - \overline{q}\|_{\infty} \leq 1\,.
    \end{equation}

    
    We now argue that \eqref{eq:p-q-norm-leq-1} implies a similar inequality on $\|p - \widetilde{q}\|_{\infty}$. Let $\mathbbm{1} = (1, \ldots, 1)$ and let $\theta$ be the angle between $\mathbbm{1}$ and $(p-u)$; since $(p-u)$ and $(\overline{p} - u)$ are parallel, this is also the angle between $(\overline{p} - u)$ and $\mathbbm{1}$. Because $(p-u)$ is parallel to $(\overline{p} - u)$:
    \begin{align}
        p &= u + (p - u) \\
        &= u + (\overline{p} - u) \frac{\|p-u\|_2 \|\mathbbm{1}\|_2 \cos \theta}{\|\overline{p} - u\|_2 \|\mathbbm{1}\|_2 \cos \theta} \explain{Since $(p-u)$ is parallel to $(\overline{p}-u)$} \\
        &= u + (\overline{p} - u) \frac{\langle p-u, \mathbbm{1} \rangle}{\langle \overline{p} - u, \mathbbm{1} \rangle} \\
        &= u + (\overline{p} - u) \frac{wt(p-u)}{wt(\overline{p}-u)} \label{eq:p-in-terms-of-pbar}
    \end{align}
    Combining \eqref{eq:q-tilde-in-terms-of-pbar} and \eqref{eq:p-in-terms-of-pbar} gives:
    \begin{align}
        \|p - \widetilde{q}\|_{\infty} &= \left\|(\overline{p}-\overline{q})\frac{wt(p-u)}{wt(\overline{p}-u)} \right\|_{\infty} = \|\overline{p} - \overline{q}\|_{\infty} \frac{wt(p-u)}{wt(\overline{p}-u)} \notag \\
        & = \|\overline{p} - \overline{q}\|_{\infty} \frac{b-a}{wt(\overline{p}-u)} \explain{Since $p \in B$, so $wt(p)=b$, and $u \in Low_a$, so $wt(u)=a$}\\
        &\leq 1 \cdot \frac{b-a}{\left(\frac{11}{24}(b-a)\right)} \explain{By \eqref{eq:big_difference_wt_x_and_u} applied to $x=\overline{p}$ and \eqref{eq:p-q-norm-leq-1}} \\
        &= \frac{24}{11} \,. \label{eq:diff_p_q_tilde_at_most_3}
    \end{align}

    We now show that $\|q - \widetilde{q}\|_{\infty}$ is small. Intuitively, since $p \in B$, by \eqref{eq:diff_p_q_tilde_at_most_3} we have that $\widetilde{q}$ is close to $B$; also since $q \in B$ and $q$ is collinear with $\widetilde{q}$ and $u$, they must be close to each other.
    To formally show this, consider $(\overline{q} - u)$, which is parallel to $(q - \widetilde{q})$. We define $\overline{B}$ as the plane parallel to $B$ through $\overline{q}$:
    \begin{align}
        \overline{B} = \Bigl\{\vec{x} \in \mathbb{R}^k \mid x_1 + \ldots + x_k = wt(\overline{q})\Bigr\}\,.
    \end{align}
    Let $z$ be the projection of $u$ onto $\overline{B}$ in the direction $\mathbbm{1}$. 

    We have:
    \begin{align}
        wt(z) - wt(u) & = wt(\overline{q}) - wt(u) \explain{Because $z \in \overline{B}$}\\
        & \geq \frac{11}{24}(b-a) \explain{By applying \eqref{eq:big_difference_wt_x_and_u} to $\overline{q}$.}\\
        \label{eq:wt_z_minus_wt_u_lb}
    \end{align}

    Since $z - u$ has slope $(1, \ldots, 1)$ by definition of $z$, the coordinates of $z-u$ are all the same. Therefore for all $i \in [k]$:
    \begin{align}
        (z-u)_i & \geq \frac{11}{24k}(b-a) \explain{Using \eqref{eq:wt_z_minus_wt_u_lb}}  \\
        & \geq \frac{11}{2}L  \label{eq:q_minus_u_i_lb}
    \end{align}

    Define the continuous version $T^C_L$ of the tube $T_L$ as:
    \begin{align}
        T^C_L = \Bigl\{\vec{x} \in \mathbb{R}^k \mid \exists \alpha \in \mathbb{R} \mbox{ such that } |x_j - \alpha| \leq L\  \forall j \in [k] \Bigr\} \label{eq:continuous-tube-definition}
    \end{align}
    Let $\vec{y} \in T_L^C \cap \overline{B}$ be arbitrary, and let $\alpha_y$ be its value of $\alpha$ from \eqref{eq:continuous-tube-definition}. Because $\vec{y} \in \overline{B}$, we have the pair of inequalities:
    \begin{align}
        wt(\overline{q}) &= \sum_{i=1}^k y_i \leq k (\alpha_y + L) \label{eq:wt-y-sum-upper-bound} \\
        wt(\overline{q}) &= \sum_{i=1}^k y_i \geq k (\alpha_y - L) \label{eq:wt-y-sum-lower-bound}
    \end{align}
    Solving each of \eqref{eq:wt-y-sum-lower-bound} and \eqref{eq:wt-y-sum-upper-bound} for $\alpha_y$ gives $|\alpha_y - (wt(\overline{q})/k)| \leq L$. As $\vec{y} \in T_L^C \cap \overline{B}$ was arbitrary, this bound applies to the corresponding values $\alpha_{\overline{q}}$ and $\alpha_z$ for $\overline{q}$ and $z$:
    \begin{align}
        \left| \alpha_{\overline{q}} - \frac{wt(\overline{q})}{k} \right| &\leq L \label{eq:alpha-q-close-to-wt-q-over-k} \\
        \left| \alpha_{z} - \frac{wt(\overline{q})}{k} \right| &\leq L \label{eq:alpha-z-close-to-wt-q-over-k}
    \end{align}
    We now get:
    \begin{align}
        \|\overline{q} - z\|_{\infty} &\leq (\max (\alpha_{\overline{q}}, \alpha_z) + L) - (\min (\alpha_{\overline{q}}, \alpha_z) - L) \explain{As $\overline{q}, z \in T^C_L$}\\
        &= |\alpha_{\overline{q}} - \alpha_z| + 2L \\
        &\leq L + L + 2L \explain{By the triangle inequality with \eqref{eq:alpha-q-close-to-wt-q-over-k} and \eqref{eq:alpha-z-close-to-wt-q-over-k}} \\
        &= 4L \label{eq:qbar-minus-z-small}
    \end{align}

\begin{figure}[h!]
        \centering 
        \includegraphics[scale=0.6]{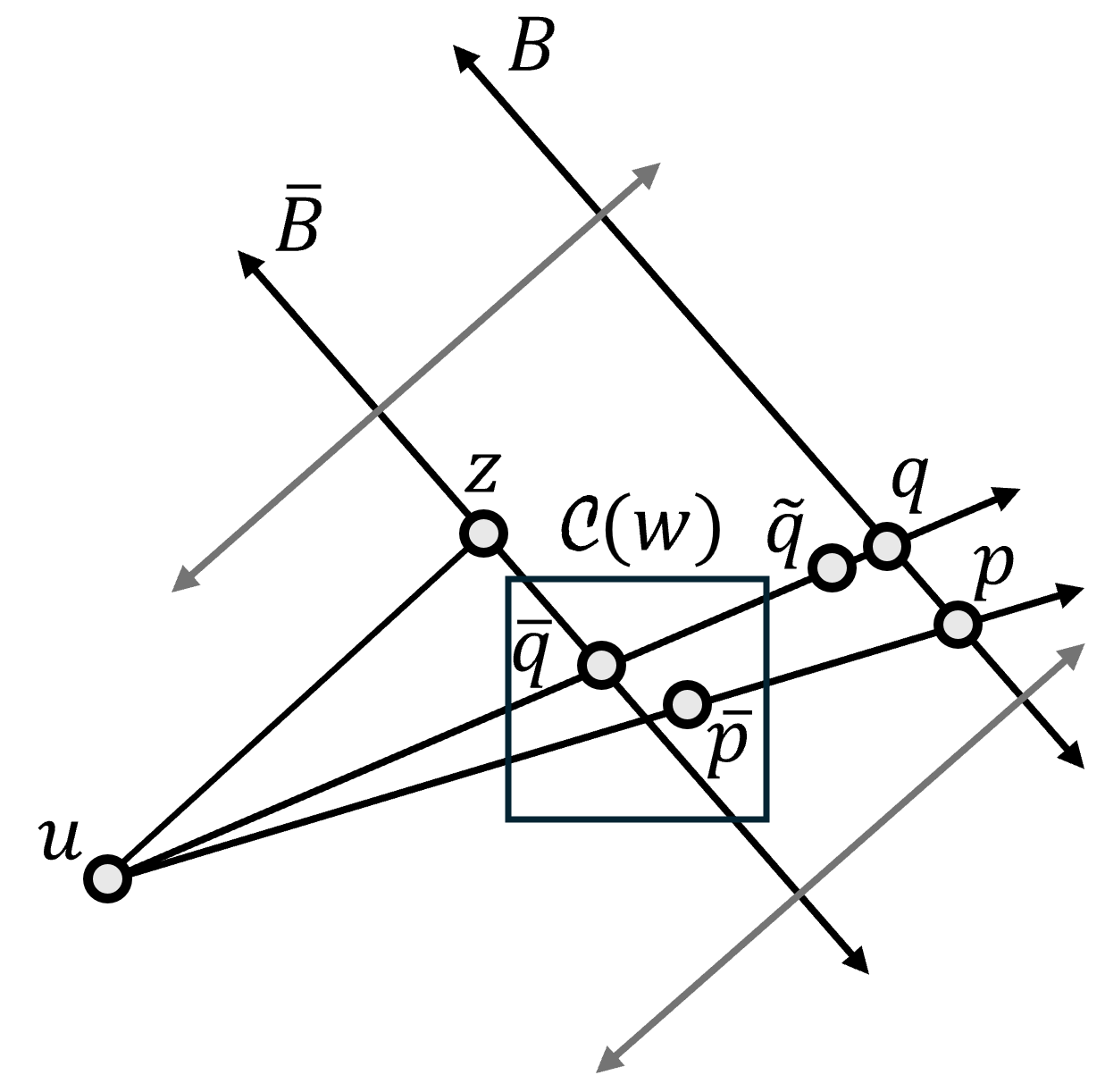}
        \caption{The geometry of Lemma \ref{lem:random-spine-unlikely-to-hit-vertex}. In addition to what is pictured in Figure \ref{fig:random-spine-lemma-geometry-lite}, this diagram includes the plane $\overline{B}$, which is parallel to $B$ and goes through $\overline{q}$. The point $z \in \overline{B}$ is such that the segment $uz$ is perpendicular to $\overline{B}$. The unlabeled lighter lines in the top-left and bottom-right represent the edges of the tube $T_L$.}
    \end{figure}

    Since all the components of $z-u$ are the same, we have 
    \begin{align}
        \max_{i, j \in [k]} \frac{(\overline{q}-u)_i}{(\overline{q}-u)_j}  =  \max_{i, j \in [k]} \frac{(\overline{q}-z)_i + (z-u)_i}{(\overline{q}-z)_j + (z-u)_j} =  \max_{i, j \in [k]} \frac{(\overline{q}-z)_i + (z-u)_i}{(\overline{q}-z)_j + (z-u)_i} \label{eq:ratio_q_bar_minus_u_i_j}
\end{align}
Inequality \eqref{eq:q_minus_u_i_lb} gives $(z-u)_i \geq \frac{11}{2}L$, so 
\begin{align}
    -4L + (z-u)_i > 0\,. \label{eq:minus_2L_plus_z_minus_u_i}
\end{align}
We obtain 
\begin{align} 
        \max_{i, j \in [k]} \frac{(\overline{q}-z)_i + (z-u)_i}{(\overline{q}-z)_j + (z-u)_i}  & \leq \max_{i, j \in [k]} \frac{4L + (z - u)_i}{-4L + (z-u)_i} \explain{By \eqref{eq:qbar-minus-z-small} and \eqref{eq:minus_2L_plus_z_minus_u_i}} \\
        & \leq \frac{\frac{11}{2}L + 4L}{\frac{11}{2}L - 4L} \explain{By \eqref{eq:q_minus_u_i_lb}} \\
        &= \frac{19}{3}\,. \label{eq:max_i_j_ratio_q_minus_z_over_z_minus_u_at_most_9}
    \end{align}

    Combining \eqref{eq:ratio_q_bar_minus_u_i_j}  and \eqref{eq:max_i_j_ratio_q_minus_z_over_z_minus_u_at_most_9} yields 
    \begin{align}\label{eq:qbar-minus-u-coordinate-ratio}
        \max_{i, j \in [k]} \frac{(\overline{q}-u)_i}{(\overline{q}-u)_j} \leq \frac{19}{3}\,.
    \end{align}

    Because $\overline{q}-u$ and $q - \widetilde{q}$ are parallel, we have $(q-\widetilde{q}) = \beta (\overline{q}-u)$ for some $\beta \in \mathbb{R}$. Therefore applying \eqref{eq:qbar-minus-u-coordinate-ratio} gives 
    \begin{align}\label{eq:q-minus-qtilde-coordinate-ratio}
        \max_{i, j \in [k]} \frac{(q-\widetilde{q})_i}{(q-\widetilde{q})_j} = \max_{i, j \in [k]} \frac{\beta (\overline{q}-u)_i}{\beta (\overline{q}-u)_j} \leq \frac{19}{3} \,.
    \end{align}
    We also have 
    \begin{align}
    |wt(\widetilde{q}) - b| &= \left| \sum_{i=1}^k \widetilde{q}_i - p_i \right| \explain{Since $p \in B$, we have $\sum_{i=1}^k p_i=b$} \\
    &\leq \sum_{i=1}^k \left| \widetilde{q}_i - p_i\right| \explain{By triangle inequality} \\
    &\leq \frac{24}{11}k\,. \explain{Since $\|p - \widetilde{q}\|_{\infty} \leq 24/11$ by \eqref{eq:diff_p_q_tilde_at_most_3}}  \\
    \label{eq:up_wt_q_tilde_minus_b_absolute}
    \end{align}
    We define the unit-weight vector in the direction of $q-\widetilde{q}$ as
    \begin{align} d = \frac{q-\widetilde{q}}{wt(q-\widetilde{q})} \,.
    \end{align}
    We then have:
    \begin{align}
        \max_{i \in [k]} d_i &= \left(\min_{j \in [k]} d_j\right) \cdot \left(\max_{i, j \in [k]} \frac{d_i}{d_j}\right) \\
        &\leq \frac{1}{k} \cdot \left(\max_{i, j \in [k]} \frac{d_i}{d_j}\right) \\
        &\leq \frac{19}{3k} \explain{Since $d$ is parallel to $q-\widetilde{q}$, so the bound of \eqref{eq:q-minus-qtilde-coordinate-ratio} applies}\\
        \label{eq:9_over_k_bound}
    \end{align}
    Therefore:
    \begin{align}
        \|q - \widetilde{q}\|_{\infty} &= \|d\|_{\infty} \cdot |wt(\tilde{q}) - b| \\
        &\leq \frac{19}{3k} \cdot \frac{24}{11}k \explain{By \eqref{eq:9_over_k_bound} and \eqref{eq:up_wt_q_tilde_minus_b_absolute}}\\
        &= \frac{152}{11}  \,. \label{eq:q_minus_q_tilde_norm_ub}
    \end{align}

    Inequality \eqref{eq:diff_p_q_tilde_at_most_3} gives $\|p - \widetilde{q}\|_{\infty} \leq 24/11$. Since  $\|q - \widetilde{q}\|_{\infty} \leq 152/11$ by \eqref{eq:q_minus_q_tilde_norm_ub}, applying the triangle inequality gives   
    \begin{equation}
        \|p - q\|_{\infty} \leq \|p - \widetilde{q}\|_{\infty}  + \|q - \widetilde{q}\|_{\infty} \leq \frac{24}{11} + \frac{152}{11} = 16\,.
    \end{equation}

    Since this bound applies to any two points $p, q \in \mathcal{P}(w)$, all of $\mathcal{P}(w)$ fits inside a hypercube of side length $16$. This hypercube contains at most $(16+1)^k$ lattice points, so  at most $17^k$ realizations of $V$ could result in $\vec{s}(u, V)$ passing through $w$. Since the total number of realizations of $V$ is $|Low_b|$, we have 
    \begin{align}
        \Pr\Bigl[w \in \vec{s}(u,V)\Bigr] &\leq \frac{17^k}{|Low_b|} \\
        &\leq \frac{17^k}{(2L+1)^{\lfloor k/2 \rfloor}}\,. \explain{By Lemma \ref{lem:size_Low_a}}
    \end{align}
    This bound applies to any fixed $u \in Low_a$, so it also holds when $u$ is chosen uniformly at random from $Low_a$. Thus for $U \sim \mathcal{U}(Low_a)$ and $V \sim \mathcal{U}(Low_b)$ we have 
    \begin{align}
        \Pr\Bigl[w \in \vec{s}(U, V)\Bigr] &\leq \frac{17^k}{(2L+1)^{\lfloor k/2 \rfloor}},
    \end{align}
    as required.
\end{proof}

\begin{lemma}\label{lem:rv-max-less-than-kc}
    Let $X_1, \ldots, X_m$ be nonnegative, possibly correlated random variables. Let $C \in \mathbb{R}$ be such that $\mathbb{E}[X_i] \leq C$ for all $i \in [m]$. Then:
    \begin{align}
        \mathbb{E}\left[\max_{i \in [m]} X_i\right] &\leq mC \,.
    \end{align}
\end{lemma}

\begin{proof}
    Because the $X_i$ are nonnegative, their maximum is at most their sum. Therefore:
    \begin{align}
        \mathbb{E}\left[\max_{i \in [m]} X_i\right] &\leq \mathbb{E}\left[\sum_{i=1}^m X_i\right] 
        = \sum_{i=1}^m \mathbb{E}[X_i] 
        \leq mC \,.
    \end{align}
\end{proof}

\lemmaFvsizelessthankcubed*

\begin{proof}
    Consider the roles $v$ might play for an arbitrary spine $\vec{s}$ and index $j$:
    \begin{itemize}
        \item If $v$ is the fixed point of $h^{\vec{s}, j}$, then $h^{\vec{s}, j}(v) = v$.
        \item If $v$ is on the spine $\vec{s}$ but is not the fixed point, then $h^{\vec{s}, j}(v)$ differs from $v$ in exactly one coordinate, and is either one more or one less in that coordinate. There are $2k$ ways to choose such a value.
        \item If $v$ is not on the spine $\vec{s}$, then $h^{\vec{s}, j}(v)$ differs from $v$ in exactly two coordinates: one in which $h^{\vec{s}, j}$ is one bigger, and one in which it is one smaller. There are $k(k-1)$ ways to choose such a value.
    \end{itemize}
    Therefore:
    \begin{align}
        |F(v)| &\leq 1 + 2k + k(k-1) = k^2 + k + 1 \leq k^3, \label{eq:size-F-v-less-than-k-cubed}
    \end{align}
    which completes the proof.
\end{proof}

\lemmamanyfarapartspineverticeshard*

\begin{proof}
In this proof, all probabilities and expectations are taken over the input function $f \sim \mathcal{U}$.

For each region $R$, define the following random variables:
\begin{itemize}
\item $P_t(R) = \max_{v \in R} \log_2 \Pr \left[\text{$v$ is a spine vertex after $t$ queries from $\mathcal{A}$}\right] $.
\item $ P^*_t(R) = \max_{0 \leq t^* \leq t} P_{t^*}(R)$.
\end{itemize}
Define the set:
\begin{align}
    G_m = \left\{(R^1, \ldots, R^m) \mid R^1, \ldots, R^m \mbox{ are regions, } dist(R^i, R^j) \geq 5 \mbox{ for all } i \neq j \in [m]\right\} \,.
\end{align}
 That is, $G_m$ is the set of $m$-tuples of regions which $\mathcal{A}$ is trying to find spine vertices in. Finally, let:
\begin{align}
    \overline{P}_t &= \max_{(R^1, \ldots, R^m) \in G_m} \sum_{i=1}^m P^*_t(R^i) \,. 
\end{align}
We aim to upper-bound $\mathbb{E}[\overline{P}_{t+1} - \overline{P}_t]$. To do so, we condition on the history $\mathcal{H}_t$ of the first $t$ queries made by $\mathcal{A}$ and their responses.

As in Lemma \ref{lem:F-v-size-less-than-k-cubed}, for each vertex $v$, let $F(v)$ be the set of possible responses to querying $v$: 
\begin{align}
    F(v) &= \left\{y \in \{0, \ldots, n-1\}^k \mid \exists \vec{s} \in \Psi, j \in \{0, \ldots, k(n-1)\} \mbox{ such that } h^{\vec{s}, j}(v) = y\right\} \,. \notag 
\end{align}

Consider an arbitrary region $R$, and let $q_{t+1}$ be the next query made by $\mathcal{A}$. Let $V_t(R)$ be the set of vertices $v \in R$ which could still be spine vertices:
\begin{align}
    V_t(R) = \left\{v \in R \mid \Pr[\text{$v$ is a spine vertex} \mid \mathcal{H}_t] > 0\right\} \,.
\end{align}
We now upper-bound $\mathbb{E}[P^*_{t+1}(R) - P^*_t(R) \mid \mathcal{H}_t]$. This expectation can be broken into two cases depending on the realizations of $P_{t+1}(R)$ and $P_t(R)$:
\begin{itemize}
    \item If $P_{t+1}(R) > P_t(R)$, then $P^*_{t+1}(R) - P^*_t(R) \leq P_{t+1}(R) - P_t(R)$.
    \item If $P_{t+1}(R) \leq P_t(R)$, then $P^*_{t+1}(R) = P^*_t(R)$.
\end{itemize}
In each case, $P^*_{t+1}(R) - P^*_t(R) \leq \max (0, P_{t+1}(R) - P_t(R))$. We now decompose that expected value based on the response $a_{t+1}$ to $\mathcal{A}$'s query $q_{t+1}$:
\begin{align}
    &\mathbb{E}[P^*_{t+1}(R) -P^*_t(R) \mid \mathcal{H}_t] \\
    &\leq \mathbb{E}[\max (0, P_{t+1}(R) - P_t(R)) \mid \mathcal{H}_t] \\
    &= \sum_{y \in F(q_{t+1})} \Pr \left[a_{t+1} = y \mid \mathcal{H}_t\right] \max \left(0, \log_2 \left(\frac{\max_{v \in R} \Pr[\text{$v$ is a spine vertex} \mid \mathcal{H}_t, a_{t+1} = y]}{\max_{w \in R} \Pr[\text{$w$ is a spine vertex} \mid \mathcal{H}_{t}]}\right)\right) \\
    &\leq \sum_{y \in F(q_{t+1})} \Pr \left[a_{t+1} = y \mid \mathcal{H}_t\right] \max \left(0, \max_{v \in V_t(R)} \log_2 \left(\frac{\Pr [\text{$v$ is a spine vertex} \mid \mathcal{H}_t, a_{t+1} = y]}{\Pr [\text{$v$ is a spine vertex} \mid \mathcal{H}_t]}\right)\right) \,. \label{eq:sum-f-max-log}
\end{align}

For arbitrary $v \in R$ and $y \in F(q_{t+1})$, the quantity inside the logarithm in \eqref{eq:sum-f-max-log} can be upper-bounded by expanding its denominator:
\begin{align}
    &\Pr [\text{$v$ is a spine vertex} \mid \mathcal{H}_t] \\
    &= \sum_{w \in F(q_{t+1})} \Pr \left[a_{t+1} = w \mid \mathcal{H}_t\right] \cdot \Pr [\text{$v$ is a spine vertex} \mid \mathcal{H}_t, a_{t+1} = w] \\
    &\geq \Pr \left[a_{t+1} = y \mid \mathcal{H}_t\right] \cdot \Pr [\text{$v$ is a spine vertex} \mid \mathcal{H}_t, a_{t+1} = y] \,. \label{eq:pr-f-eq-y-times-pr-v-spine}
\end{align}
Dividing both sides of \eqref{eq:pr-f-eq-y-times-pr-v-spine} by $\Pr [\text{$v$ is a spine vertex} \mid \mathcal{H}_t]$ and $\Pr \left[a_{t+1} = y \mid \mathcal{H}_t\right]$ gives:
\begin{align}
    \frac{\Pr [\text{$v$ is a spine vertex} \mid \mathcal{H}_t, a_{t+1} = y]}{\Pr [\text{$v$ is a spine vertex} \mid \mathcal{H}_t]} \leq \frac{1}{\Pr \left[a_{t+1} = y \mid \mathcal{H}_t\right]} \,. \label{eq:pr-v-spine-given-new-query-small}
\end{align}
Applying \eqref{eq:pr-v-spine-given-new-query-small} to \eqref{eq:sum-f-max-log} gives:
\begin{align}
    &\mathbb{E}[P^*_{t+1}(R) - P^*_t(R) \mid \mathcal{H}_t] \\
    &\leq \sum_{y \in F(q_{t+1})} \Pr \left[a_{t+1} = y \mid \mathcal{H}_t\right] \max \left(0, \max_{v \in V_t(R)} \log_2 \left(\frac{1}{\Pr \left[a_{t+1} = y \mid \mathcal{H}_t\right]}\right)\right) \\
    &= \sum_{y \in F(q_{t+1})} - \Pr \left[a_{t+1} = y \mid \mathcal{H}_t\right] \min \left(0, \log_2 \Pr \left[a_{t+1} = y \mid \mathcal{H}_t\right]\right) \\
    &\leq \sum_{y \in F(q_{t+1})} - \Pr \left[a_{t+1} = y \mid \mathcal{H}_t\right] \log_2 \Pr \left[a_{t+1} = y \mid \mathcal{H}_t\right] \,. \label{eq:expected-sr-gain-less-than-entropy}
\end{align}
Because \eqref{eq:expected-sr-gain-less-than-entropy} is precisely the entropy of $a_{t+1}$ given $\mathcal{H}_t$, it can be upper-bounded by the entropy of a uniform distribution over $F(q_{t+1})$. Further applying Lemma \ref{lem:F-v-size-less-than-k-cubed}:
\begin{align}
    \mathbb{E}[P^*_{t+1}(R) - P^*_t(R) \mid \mathcal{H}_t] &\leq \log_2(|F(q_{t+1})|) \leq 3\log_2(k) \label{eq:Pt-change-given-history-small}
\end{align}

We now bound $\mathbb{E} [\overline{P}_{t+1} - \overline{P}_t \mid \mathcal{H}_t]$. This gain is upper-bounded by the expected ``most improved" tuple in $G_m$:
\begin{align}
    \mathbb{E} [\overline{P}_{t+1} - \overline{P}_t \mid \mathcal{H}_t] &\leq \mathbb{E} \left[ \max_{(R^1, \ldots, R^m) \in G_m} \sum_{i=1}^m P^*_{t+1}(R^i) - P^*_t(R^i) \mid \mathcal{H}_t\right] \,.  \label{eq:p-star-gain-less-than-most-improved-tuple}
\end{align}

Each tuple in $G_m$ contains only regions that are at least five apart. By Lemma \ref{lem:query-only-affects-adjacent-regions}, querying $q_{t+1}$ can only update knowledge about the spine in a block of five consecutive regions; call this block $U(q_{t+1})$. Therefore, each individual sum in \eqref{eq:p-star-gain-less-than-most-improved-tuple} contains at most one nonzero term, and all such terms are for one of the five regions in $U(q_{t+1})$. Therefore:
\begin{align}
    \mathbb{E} [\overline{P}_{t+1} - \overline{P}_t \mid \mathcal{H}_t] &\leq \mathbb{E} \left[ \max_{R \in U(q_{t+1})} P^*_{t+1}(R) - P^*_t(R) \mid \mathcal{H}_t\right] \,. 
\end{align}
Because $P^*_{t+1}(R) - P^*_t(R) \geq 0$ for all regions $R$, Lemma \ref{lem:rv-max-less-than-kc} combined with \eqref{eq:Pt-change-given-history-small} applies to give:
\begin{align}
    \mathbb{E} [\overline{P}_{t+1} - \overline{P}_t \mid \mathcal{H}_t] &\leq 15 \log_2(k) \,. \label{eq:Pstar-change-less-than-15-log-k-with-history}
\end{align}
And as the upper bound in \eqref{eq:Pstar-change-less-than-15-log-k-with-history} applies independently of $\mathcal{H}_t$, taking the expectation over $\mathcal{H}_t$ gives the history-independent bound:
\begin{align}
    \mathbb{E} [\overline{P}_{t+1} - \overline{P}_t] &\leq 15 \log_2(k)\,.  \label{eq:Pstar-change-less-than-15-log-k}
\end{align}
Now suppose an algorithm $\mathcal{A}$ makes at most $T$ queries. By telescoping \eqref{eq:Pstar-change-less-than-15-log-k}, we get:
\begin{align}
    \mathbb{E} [\overline{P}_T - \overline{P}_0] &\leq 15 T \log_2(k) \,.
\end{align}
When $\mathcal{A}$ succeeds, it has found $m$ spine vertices in regions at least five apart, and therefore $\overline{P}_T = 0$. By Lemma \ref{lem:random-spine-unlikely-to-hit-vertex}, no vertex other than those in regions $R_0$ and $R_{k(n-1)-\rho}$ initially has a chance higher than $17^k / (2L+1)^{\lfloor k/2 \rfloor}$ to be on the spine. In those regions, we can vacuously upper-bound $P^*_0(R_0)$ and $P^*_0(R_{k(n-1)-\rho})$ by $0$, so:
\begin{align}
    \overline{P}_0 &\leq (m-2) \log_2 \left(\frac{17^k}{(2L+1)^{\lfloor k/2 \rfloor}}\right) \,. 
\end{align}
Therefore, by Markov's inequality, the probability $\mathcal{A}$ succeeds within $T$ queries is at most:
\begin{align}
    \Pr [\text{$\mathcal{A}$ succeeds in $T$ queries}] &\leq \frac{15 T \log_2(k)}{(m-2) \cdot \left(\lfloor k/2 \rfloor \log_2(2L+1) - k \log_2 (17)\right)} \,. 
\end{align}
This completes the proof.
\end{proof}

We now present Lemma \ref{lem:many-region-queries-required}, which states that solving $TARSKI(n, k)$ requires querying spine vertices in many different regions. Its proof is based around the classical ordered search problem: 

\begin{definition}[Ordered search] \label{def:ordered_search_problem}
		The  input is  a bit vector $\vec{x} = (x_1, \ldots, x_m) \in \{0,1\}^m$ with the promise that exactly one bit is set to $1$. The vector can be accessed  via oracle queries of the form: ``Is the $i$-th bit equal to $1$?''. The answer to a query is: ``Yes'', ``No, go left'', or ``No, go right''.
		The task is to find the location of the hidden bit. 
\end{definition}

\lemmamanyregionqueriesrequired* 

The proof of Lemma \ref{lem:many-region-queries-required} depends on lower bounds for ordered search, which we present afterwards.

\begin{proof}
    Whenever $\mathcal{A}$ correctly returns the fixed point, it also learns which region the fixed point is in. We will argue that even learning the region with the fixed point is at least as hard as ordered search on the $m$ regions, where $m = k(n-1) / \rho$.

    To that end, let $\mathcal{A}'$ be the following randomized algorithm for the ordered search problem of \cref{def:ordered_search_problem}, which is given an input $\vec{x} \in \{0, 1\}^m$ with (unknown) answer $j$.
\begin{itemize}
    \item Draw a spine $\vec{t} \in \Psi$ and an offset $\theta \in \{0, \ldots, \rho - 1\}$ uniformly at random. 
    \item Set variables $a = 0$ and $b = k(n-1)$.
    \item Simulate $\mathcal{A}$ on $h^{\vec{t}, \rho j + \theta}$. While $\mathcal{A}'$ does not know $j$, it can still answer each query $v$ submitted by $\mathcal{A}$ as follows:
    \begin{itemize}
        \item If $v \notin \vec{t}$, then compute $h^{\vec{t}, q}(v)$ for an arbitrary $q \in \{0, \ldots, k(n-1)\}$.
        \item If $v \in \vec{t}$, then let $R_{\rho i}$ be the region containing $v$.  Let $\mathcal{A}'$ query the vector $\vec{x}$ at position $i$; if it had already done so, then look up the result of that query instead. Then:
        \begin{itemize}
            \item If the response is ``No, go left", then set $b = \min(b, wt(v)-1)$ and give $\mathcal{A}$ the value of $h^{\vec{t}, q}(v)$ for an arbitrary $q \in [a, b]$.
            \item If the response is ``No, go right", then set $a = \max(a, wt(v)+1)$ and give $\mathcal{A}$ the value of $h^{\vec{t}, q}(v)$ for an arbitrary $q \in [a, b]$.
            \item If the response is ``Yes", then halt; the answer has been found.
        \end{itemize}
    \end{itemize}
    \item Whenever $\mathcal{A}$ outputs $u$ as its answer, $\mathcal{A}'$ returns the index $i$ such that $u \in R_{\rho i}$.
\end{itemize}

We first argue by induction on the number of queries $\mathcal{A}$ makes that:
\begin{itemize}
    \item[(i)] $\mathcal{A}'$ can accurately simulate $h^{\vec{t}, \rho j + \theta}$, given its choice of $\vec{t}$ and $\theta$.
    \item[(ii)] $\mathcal{A}'$ maintains $\rho j + \theta \in [a, b]$.
\end{itemize}
The base case is clear, as $\rho j + \theta \geq 0$ and $\rho j + \theta \leq k(n-1)$.

Now suppose that $\rho j + \theta \in [a, b]$ after some number of queries from $\mathcal{A}$. Its next query $v$ falls into one of a few cases:
\begin{itemize}
    \item $v \notin \vec{t}$. Then by definition, the value of $h^{\vec{t}, \rho j + \theta}(v)$ is independent of $\rho j + \theta$, so $\mathcal{A}'$ can compute it.
    \item $v \in \vec{t}$. Then there is a region $R_{\rho i}$ such  that $v \in R_{\rho i}$, and $\mathcal{A}'$ queries bit $i$ of $\vec{x}$. There are three cases for the resulting query:
    \begin{itemize}
    \item If $\mathcal{A}'$ receives ``No, go left", then $i > j$. Therefore, $wt(v) \geq \rho (j+1)$, so after updating $b$ we maintain $b \geq \rho j + \theta$. The value of $h^{\vec{t}, \rho j + \theta}(v)$ can then be determined from knowing $\vec{t}$ and that $wt(v) > \rho j + \theta$.
    \item If $\mathcal{A}'$ receives ``No, go right", then $i < j$. Therefore, $wt(v) < \rho j$, so after updating $a$ we maintain $a \leq \rho j + \theta$. The value of $h^{\vec{t}, \rho j + \theta}(v)$ can then be determined from knowing $\vec{t}$ and that $wt(v) < \rho j + \theta$.
    \item If $\mathcal{A}'$ receives ``Yes", then it halts immediately, having found its answer of $j$.
    \end{itemize}
\end{itemize}

Therefore, by induction, each response $\mathcal{A}$ receives is consistent with $h^{\vec{t}, \rho j + \theta}$.

If $\mathcal{A}'$ is run on an input distribution which is uniformly random over the $m$ possible answers, then its simulated values of both $\mathbf{t} \in \Psi$ and $\rho j + \theta \in \{0, \ldots, k(n-1)\}$ will be uniformly random as well. Therefore, its simulation of $\mathcal{A}$ will be run on inputs from $\mathcal{U}$.

Now suppose $\mathcal{A}$ is correct with probability at least $1-\delta$ on inputs drawn from $\mathcal{U}$. Then $\mathcal{A}'$ will be correct with probability at least $1-\delta$ on uniformly random inputs. By Lemma \ref{lem:ordered-search-lower-bound}, there exists $c = c(\delta, \epsilon)$ such that on this distribution $\mathcal{A}'$ makes at least $c \log m$ queries with probability at least $\epsilon$. Because each query $\mathcal{A}'$ makes corresponds to a new region queried by $\mathcal{A}$, we have that $\mathcal{A}$ queries at least $c \log m$ regions with probability at least $\epsilon$.
\end{proof}

\begin{lemma}
    \label{lem:ordered-search-lower-bound-deterministic}
    Let $\epsilon, \delta \in (0, 1)$ such that $\delta + \epsilon < 1$. Let $U_m$ be the uniform distribution over ordered search instances of length $m$. Let $\mathcal{A}$ be a deterministic algorithm for ordered search with the property that
     $    \Pr_{\vec{x} \sim U_m}\left[ \mathcal{A} \mbox{ succeeds on } \vec{x} \right] \geq 1-\delta\,. $  
    Then there exists $c > 0$ such that, for $m > 2^{3/(1-\delta-\epsilon)}$: 
    \begin{align}
        \Pr_{\vec{x} \sim U_m} \left[\mathcal{A} \mbox{ issues at least } c \log_2(m) \mbox{ queries on } \vec{x} \right] \geq \epsilon\,.\label{eq:deterministic-ordered-search-many-queries}
    \end{align}
\end{lemma}

\begin{proof}
    We will use $c = 1/6$. For contradiction, suppose $\mathcal{A}$ violated inequality \eqref{eq:deterministic-ordered-search-many-queries}. Then consider the following encoding scheme for an integer $z \in [m]$:
    \begin{itemize}
        \item If $\mathcal{A}$ succeeds within $c \log_2(m)$ queries on the ordered search instance of length $m$ with answer $z$, then encode $z$ with a $0$ followed by the responses to the queries made by $\mathcal{A}$ on this input. This can be done in $2$ bits per query, for a total of at most $1 + 2c \log_2(m)$ bits.
        \item Otherwise, encode $z$ with a $1$ followed by its representation in binary. This can be done in $1 + \lceil \log_2(m) \rceil$ bits.
    \end{itemize}
    The probability of the first case occurring on a uniformly random input is at least $1-\delta-\epsilon$, so overall the average number of bits this scheme takes is at most:
    \begin{align}
        & 1 + (1-\delta-\epsilon) 2c \log_2(m) + (\delta + \epsilon) \lceil \log_2 m \rceil \\
        \leq& 1 + (1-\delta-\epsilon) \frac{1}{3} \log_2(m) + (\delta + \epsilon) \log_2 (m) + 1 \\
        =& 2 + \frac{2\delta + 2\epsilon + 1}{3} \log_2(m) \\
        <& \frac{2-2\delta-2\epsilon}{3} \log_2(m) + \frac{2\delta + 2\epsilon + 1}{3} \log_2(m) \explain{Because $m>2^{3/(1-\delta-\epsilon)}$} \\
        =& \log_2(m) \,.
    \end{align}
    But $\log_2(m)$ is the entropy of a uniformly distributed value in $[m]$, so this violates the source coding theorem. Therefore, $\mathcal{A}$ must satisfy \eqref{eq:deterministic-ordered-search-many-queries}, which completes the proof.
\end{proof}

\begin{lemma} \label{lem:ordered-search-lower-bound}
Let $\epsilon, \delta \in (0,1)$ such that $2\delta + 2\epsilon < 1$. Let $U_m$ be the uniform distribution over ordered search instances of length $m$. Let $\mathcal{A}$ be a randomized algorithm for ordered search with the property that 
\begin{align}
    \Pr_{\vec{x} \sim U_m}\left[ \mathcal{A} \mbox{ succeeds on } \vec{x} \right] \geq 1-\delta\,. \notag 
\end{align}
Then there exists $c = c(\delta, \epsilon) > 0$ such that, for $m > 2^{3/(1-2\delta-2\epsilon)}$: 
\begin{align}
    \Pr_{\vec{x} \sim U_m} \left[\mathcal{A} \mbox{ issues at least } c \log(m) \mbox{ queries on } \vec{x} \right] \geq \epsilon\,.
\end{align}
\end{lemma}
\begin{proof}
    We can view $\mathcal{A}$ as a distribution $D_{\mathcal{A}}$ over deterministic algorithms. Because the choice of deterministic algorithm occurs before execution, the draw from $D_{\mathcal{A}}$ is independent of the random input. Then by Markov's inequality:
    \begin{align}
        \Pr_{A \sim D_{\mathcal{A}}} \left[ \Pr_{\vec{x} \sim U_m} \left[ \text{$A$ fails on $\vec{x}$}\right] > 2\delta \right] &\leq \frac{\mathbb{E}_{A \sim D_{\mathcal{A}}} \left[\Pr_{\vec{x} \sim U_m} \left[ \text{$A$ fails on $\vec{x}$}\right]\right]}{2\delta} \\
        &= \frac{\Pr_{\vec{x} \sim U_m} \left[ \text{$\mathcal{A}$ fails on $\vec{x}$}\right]}{2\delta} \\
        &\leq \frac{1}{2}\,. \label{eq:A-big-failure-less-than-one-half}
    \end{align}
    We will use the same $c$ as Lemma \ref{lem:ordered-search-lower-bound-deterministic}. By that lemma, for any deterministic algorithm $A$ with success probability at least $1-2\delta$ on the uniform distribution:
    \begin{align}
        \Pr_{\vec{x} \sim U_m} \left[A \mbox{ issues at least } c \log_2(m) \mbox{ queries on } \vec{x} \right] \geq 2\epsilon\,. \label{eq:deterministic-A-many-queries-more-than-2-epsilon}
    \end{align}
    Combining \eqref{eq:A-big-failure-less-than-one-half} and \eqref{eq:deterministic-A-many-queries-more-than-2-epsilon}:
    \begin{align}
        &\Pr_{\vec{x} \sim U_m} \left[\mathcal{A} \mbox{ issues at least } c \log(m) \mbox{ queries on } \vec{x} \right] \\
        \geq& \Pr_{A \sim D_{\mathcal{A}}} \left[\Pr_{\vec{x} \sim U_m} \left[ \text{$A$ succeeds on $\vec{x}$} \right] \geq 1-2\delta\right] \\
        &\cdot \mathbb{E}_{A \sim D_{\mathcal{A}}}  \left[\Pr_{\vec{x} \sim U_m} \left[A \mbox{ issues at least } c \log(m) \mbox{ queries on } \vec{x} \right] \mid \Pr_{\vec{x} \sim U_m} \left[ \text{$A$ succeeds on $\vec{x}$} \right] \geq 1-2\delta\right] \notag \\
        \geq& \frac{1}{2} \cdot 2\epsilon \\
        =& \epsilon \,.
    \end{align}
    This completes the proof.
\end{proof}

\section{An $O(k \log n \log(nk))$ upper bound for multi-dimensional herring-bone functions} \label{app:upper_bound_distribution_U}

In this section we include the proofs of statements in  \cref{sec:upper_bound_distribution_U}.

\theoremherringboneupperbound*

\begin{proof}
    By Lemma \ref{lem:general-herringbone-spine-find-Hamming}, the spine vertex with any particular Hamming weight can be found in $O(k \log n)$ queries. This subroutine can be used to run binary search to locate the fixed point along the spine.

    Specifically, suppose the (unknown) instance is $h^{\vec{s}, j}$ for some $\vec{s} \in \Psi$ and $j \in \{0, \ldots, (n-1)k\}$. For an arbitrary $x \in \{0, \ldots, (n-1)k\}$, querying the spine vertex $s^x$ produce three different options depending on how $x$ compares to $j$:
    \begin{itemize}
        \item If $x=j$, then $h^{\vec{s}, j}(s^x) = s^x$. This is the fixed point.
        \item If $x < j$, then $h^{\vec{s}, j}(s^x) = s^{x+1}$. This is identifiable because $h^{\vec{s}, j}(s^x) > s^x$ only in this case.
        \item If $x > j$, then $h^{\vec{s}, j}(s^x) = s^{x-1}$. This is identifiable because $h^{\vec{s}, j}(s^x) < s^x$ only in this case.
    \end{itemize}
    This is enough feedback to run binary search to find $j$. Binary search takes $O(\log (nk))$ iterations to identify $j \in \{0, \ldots, (n-1)k\}$, so overall this algorithm uses $O(k \log(n) \log(nk))$ queries.
\end{proof}

\begin{lemma}
    \label{lem:general-herringbone-spine-find-Hamming}
    Given a value $0 \le x \le (n-1)k$, there is an $O(k\log(n))$-query algorithm to find the spine vertex $s^x = (s^x_1, \ldots, s^x_k)$ that is Hamming distance $x$ from the origin in a multi-dimensional herringbone function.
\end{lemma}
\begin{proof}
    Let $a^1_i = 0$, let $b^1_i = x/k$, and let $c^1_i = n-1$ for all $i \in [k]$.
    We then perform the following steps repeatedly until a spine vertex is returned.
    Let the iteration number be $m$, starting with $1$.
    We use the notation $a^m = (a^m_1, \ldots, a^m_k)$ for all vertices $a^m$.
    \begin{enumerate}
        \item Query $f(b^m)$.
        \item If $\|f(b^m)-b^m\| \le 1$, return $b^m$ since it must be a spine vertex.
        \item Otherwise, there are dimensions $p,q \in [k]$ such that $(f(b^m))_p - b^m_p = b^m_q - (f(b^m))_q = 1$.
        \item Set $a^{m+1} = a^m$ and $c^{m+1} = c^m$, except for $a^{m+1}_p$, which we set equal to $b^m_p$ and $c^{m+1}_q$ which we set equal to $b^m_q$.
        \item Let $d^m = \min(c^m_p-b^m_p, b^m_q-a^m_q)$. Set $b^{m+1} = b^m$, except for $b^{m+1}_p$ which we set equal to $b^m_p + d^m/2$ and $b^{m+1}_q$ which we set equal to $b^m_q - d^m/2$.
    \end{enumerate}

    We claim the above algorithm maintains the following invariants for all $m$:
    \begin{enumerate}
        \item $a^m \le b^m \le c^m$.
        \item $a^m \le s^x \le c^m$.
        \item $b^m$ is Hamming distance $x$ from the origin.
    \end{enumerate}
    The invariants are initially true by the initialization of $a^1, b^1, c^1$.
    Suppose inductively that the invariants hold for $a^m, b^m, c^m$.
    We will show they hold for $a^{m+1}, b^{m+1}, c^{m+1}$ if the algorithm did not return in or before iteration $m$.
    
    Assuming the algorithm did not return in iteration $m$, we have dimensions $p,q$ such that $(f(b^m))_p - b^m_p = b^m_q - (f(b^m))_q = 1$.
    Then $b^m_p = s^{\mu(b^m)}_p$.
    Furthermore, $s^x \ge s^{\mu(b^m)}$ since all spine vertices are fully ordered and $s^{\mu(b^m)}$ has Hamming distance from the origin strictly less than $x$.
    Therefore $s^x_p \ge b^m_p$, so $a^{m+1}_p \le s^x_p$ as required.
    By symmetric logic, $s^x_q \le c^{m+1}_q$.
    All other components of $a^{m+1} \le s^x \le c^{m+1}$ hold since $a^m \le s^x \le c^m$.

    We have $a^{m+1}_p \le b^{m+1}_p$ and $b^{m+1}_q \le c^{m+1}_q$ since $d^m \ge 0$.
    Also $b^{m+1}_p \le c^{m+1}_p$ since $d^m/2 \le c^m_p - b^m_p$ and $a^{m+1}_q \le b^{m+1}_q$ since $d^m/2 \le b^m_q - a^m_q$.
    All other components of $a^{m+1} \le b^{m+1} \le c^{m+1}$ hold since $a^m \le b^m \le c^m$.

    Finally, $b^{m+1}$ is Hamming distance $x$ from the origin because $b^m$ was Hamming distance $x$ from the origin, and $b^{m+1}$ differs only in coordinates $p$ and $q$, which are larger and smaller respectively by the same amount: $d^m/2$.

    To argue that only $O(k \log(n))$ iterations are needed, we use a potential function $\Phi$, defined as
    \begin{align}
        \Phi(m) = \sum_{i \in [k]} \phi(m,i), 
    \end{align}
    where
    \begin{align}
        \phi(m,i) = \begin{cases}
            -1 & \text{if } a^m_i = c^m_i\\
            \log(b^m_i - a^m_i) & \text{if } b^m_i - a^m_i > c^m_i - b^m_i\\
            \log(c^m_i - b^m_i) & \text{otherwise}
        \end{cases}\,.
    \end{align}
    The function $\phi(m,i)$ is well defined so long as $a^m_i \le b^m_i \le c^m_i$, which is the case due to our invariants.
    We have $\Phi(1) \le k \log(n)$.
    For all $m$ we have $\Phi(m) \ge -k$.
    We next show that $\Phi(m+1)\le \Phi(m)-1$ for all $m$ such that the algorithm did not return in or before iteration $m$.
    
    First, for all $i$, $\phi(m+1,i) \le \phi(m,i)$ since since $a^{m+1}_i \ge a^m_i$ and $c^{m+1}_i \le c^m_i$ and $b^{m+1}_i$ is weakly closer to each endpoint of $[a^{m+1}_i, c^{m+1}_i]$ than $b^m_i$ is to each endpoint of $[a^m_i, c^m_i]$.

    Second, without loss of generality let $c^m_p - b^m_p \le b^m_q - a^m_q$; the alternate case is symmetric.
    Then we claim $\phi(m+1,p) \le \phi(m,p) - 1$.

    By definition of $p$, we have $a^m_p < c^m_p$, so $\phi(m,p) \ge 0$.
    Therefore if $a^{m+1}_p = c^{m+1}_p$, then $\phi(m+1,p) = -1$ so $\phi(m+1,p) \le \phi(m,p) - 1$.
    
    Otherwise, we have $d^m = c^m_p - b^m_p$, so $c^{m+1}_p - a^{m+1}_p = d^m$ and $b^{m+1}_p = a^{m+1}_p + d^m/2$.
    Thus $\phi(m+1,p) = \log(d^m/2) = \log(d^m) - 1 \le \phi(m,p) - 1$.

    Now that we have shown $\phi(m+1,p) \le \phi(m,p) - 1$, we may conclude that $\Phi(m+1) \le \Phi(m) - 1$. Therefore the algorithm only runs for at most $k \log(n) + k \in \cO(k \log(n))$ iterations, and issues the same number of queries.
\end{proof}

\end{document}